\documentclass[cmp]{svjour}
\pdfoutput=1
\usepackage{amssymb,amsmath}
\usepackage{accents}
\usepackage{tikz}

\newcommand{\wt}{\tilde}
\newcommand{\wh}{\hat}
\newcommand{\wb}{\bar}

\renewcommand{\th}[1]{\wh{\wt{#1}}}
\newcommand{\hb}[1]{\wb{\wh{#1}}}
\newcommand{\bt}[1]{\wt{\wb{#1}}}

\newcommand{\dis}{\mathop{\mathrm{\Delta}}\nolimits}
\newcommand{\res}{\mathop{\mathrm{res}}\nolimits}

\newcommand{\CR}{\mathop{\mathrm{cr}}\nolimits}
\newcommand{\mm}{\mathop{\sf{m}}\nolimits}
\newcommand{\id}{\mathrm{id}}
\newcommand{\flm}[2]
{
  \left[
  \begin{array}{c}
  #1\\
  \vdots \\
  #2
  \end{array}
  \right]
}

\makeatletter
\renewcommand\section{\@startsection{section}{1}{\z@}%
                                  {-3.5ex \@plus -1ex \@minus -.2ex}%
                                  {2.3ex \@plus.2ex}%
                                  {\normalfont\large\bfseries}}
\makeatother

\spnewtheorem{prop}{Proposition}{\bf}{\it}
\numberwithin{equation}{section}

\title{Idempotent biquadratics, Yang-Baxter maps and birational representations of Coxeter groups}
\titlerunning{Idempotent biquadratics, Yang-Baxter maps and birational Coxeter groups}
\author{{James Atkinson$^{1,2}$}}
\authorrunning{{J. Atkinson}}
\institute{1: Faculty of Engineering and Environment, Northumbria University, Newcastle upon Tyne, UK. 2: School of Mathematics and Statistics, the University of Sydney, NSW 2006, Australia. }
\date{December 2013}

\begin{document}
\maketitle
\begin{abstract}
A transformation is obtained which completes the unification of quadrirational Yang-Baxter maps and known integrable multi-quadratic quad equations.
By combining theory from these two classes of quad-graph models we find an extension of the known integrability feature, and show how this leads subsequently to a natural extension of the associated lattice geometry.
The extended lattice is encoded in a birational representation of a particular sequence of Coxeter groups.
In this setting the usual quad-graph is associated with a subgroup of type $BC_n$, and is part of a larger and more symmetric ambient space.
The model also defines, for instance, integrable dynamics on a triangle-graph associated with a subgroup of type $A_n$, as well as finite degree-of-freedom dynamics, in the simplest cases associated with $\wt{E}_6$ and $\wt{E}_8$ affine subgroups.
Underlying this structure is a class of biquadratic polynomials, that we call idempotent, which express the trisection of elliptic function periods algebraically via the addition law.
\end{abstract}
\section{Introduction}
The point of departure in this paper is the investigation of integrable quad equations.
It will be demonstrated that the primary model in the multi-quadratic class that was discovered recently by the author in collaboration with Nieszporski \cite{AtkNie}, is related in a natural way to the primary model of the quadrirational Yang-Baxter maps that have been introduced earlier by Adler, Bobenko and Suris \cite{ABSf}.
The existence of this relationship clarifies the position of the new model with respect to known quad-graph models and transformation theory in \cite{WE,hirota-0,NQC,NC,ABS,ABS2,James,PTV,PSTV,KaNie,KaNie3,James2,JamesQ}. 
The subsequent focus of this article is on a particular repercussion of this connection, namely a new interpretation in terms of the {\it idempotent} class of biquadratic polynomials.

The idempotent biquadratic has its roots in the theory of elliptic functions; it is the formula for trisection of the periods expressed algebraically via the addition law.
Interesting properties of this formula in part correspond to integrability of the aforementioned quad-graph models, but go also beyond that, in particular leading to a natural generalisation of the quad-graph dynamics.

A practical way to understand the additional structure is in terms of the braid group representation emerging from the Yang-Baxter maps \cite{VesRev}.
In the quadrirational case the braid-type generators are self-inverse, and the representation is only of the trivial part of the pure braid group, i.e., the symmetric group.
But here we identify this as a subgroup of a larger ambient group.
The precise characterisation of this larger group is one of the main technical achievements of this paper, it is a Coxeter group with unexpectedly rich structure (cf. Figure \ref{cdd}).
For instance, in the sequence of groups formed by an increasing number of generators, the last finite case is the exceptional finite reflection group $E_6$, and the associated lattice geometry has the same combinatorics as the 27 lines on a cubic surface.
The general group structure bears closest resemblance to groups whose birational representation forms part of the basic framework in which the Painlev\'e equations and their generalisations may be understood \cite{ny1,sc,kny,10E9,tt,kmnoy}.

We proceed as follows.
Section \ref{F1Q4} introduces the aforementioned quad-graph models and establishes the new relationship between them.
The idempotent class of biquadratics is introduced in Section \ref{IB} with additional material given in Appendix \ref{H}.
Section \ref{IB2} connects the quad-graph models with the idempotents via the notion of quadrirationality.
A natural closure property of the idempotents, and an associated integrable triangle-graph dynamics, is established in Section \ref{OTC}.
The nature of the relationship between the quad-graph and triangle-graph dynamics is a quite delicate matter, motivation behind a method to unify them is given by a local consideration of the problem in Section \ref{5S}, and a constructive approach which completes the solution globally, based on vertex maps, is given in Section \ref{FH}.
The solution emerges in the form of the aforementioned birational group structure, which is interpreted from the point of view of lattice-geometry and dynamics in Section \ref{LG}.
A summary and discussion of how this kind of structure fits into the broader theory of discrete integrable systems is given in Section \ref{D}.

\section{Preliminaries, $F_I$ and $Q4^*$}\label{F1Q4}
The two relevant classes of models are the quadrirational Yang-Baxter maps and integrable multi-quadratic quad-equations.
We recall the primary model from each class and connect them by a new Miura-type transformation.

Rational mappings with the Yang-Baxter property were constructed by Adler, Bobenko and Suris \cite{ABSf} on the basis of a simpler property termed quadrirationality.
The primary model obtained was denoted $F_{I}$:
\begin{equation}
\begin{split}
\wh{u} &= \alpha v \frac{(1-\beta)u+\beta-\alpha-(1-\alpha)v}{(1-\alpha)\beta u+(\alpha-\beta)uv-(1-\beta)\alpha v},\\
\wt{v} &= \beta u \frac{(1-\alpha)v+\alpha-\beta-(1-\beta)u}{(1-\beta)\alpha v+(\beta-\alpha)vu-(1-\alpha)\beta u}.
\end{split}\label{F1}
\end{equation}
It is natural to view such models as dynamical systems on a quad-graph \cite{BS1,AdVeQ,pip}, see Figure \ref{quadgraph}(a); variables assigned to edges are governed by the system (\ref{F1}) on each quad, whilst the essential parameters of the system, $\alpha$ and $\beta$, are associated with characteristics.
This is unambiguous due to invariance of (\ref{F1}) under permutations $\wh{u}\leftrightarrow u$, $\wt{v}\leftrightarrow v$ and $(u,\wh{u},\beta)\leftrightarrow (v,\wt{v},\alpha)$, or in other words the system respects the symmetry of the quad, see Figure \ref{quadgraph}(b).
The rationality of the mapping plus the first two of these permutation symmetries is the essence of the quadrirationality.

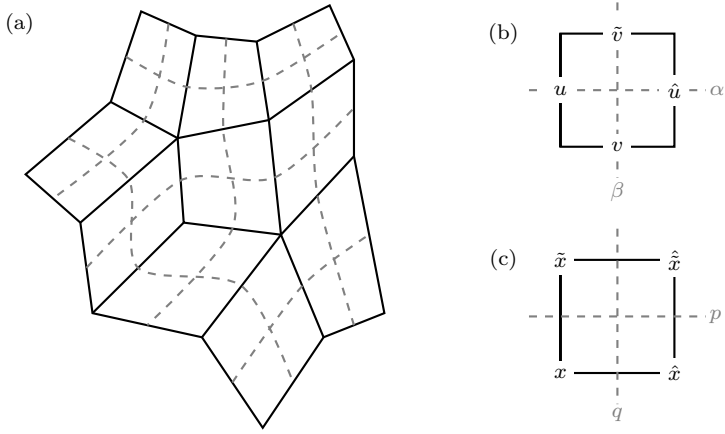
\begin{figure}[t]
\begin{center}
\begin{tikzpicture}[thick,scale=0.8]
  \tikzstyle{every node}=[inner sep=0pt]
  \coordinate (1) at (2.5,7.2); 
  \coordinate (2) at (3.4,6.8);
  \coordinate (3) at (4.4,6.7); 
  \coordinate (4) at (5.6,7.3);
  \coordinate (5) at (2.0,5.7);
  \coordinate (6) at (3.1,5.1);
  \coordinate (7) at (4.6,5.4);
  \coordinate (8) at (6.0,6.4);
  \coordinate (9) at (0.6,4.5);
  \coordinate (10) at (1.5,3.7);
  \coordinate (11) at (3.2,3.7);
  \coordinate (12) at (4.8,3.5);
  \coordinate (13) at (6.0,4.8);
  \coordinate (14) at (1.7,2.2);
  \coordinate (15) at (3.5,1.8);
  \coordinate (16) at (4.5,0.3);
  \coordinate (17) at (5.5,1.8);
  \coordinate (18) at (6.5,2.2);
  \draw (1) -- node(e12){} (2) -- node(e23){} (3) -- node(e34){} (4) -- node(e48){} (8) -- node(e813){} (13) -- node(e1318){} (18) -- node(e1718){} (17) -- node(e1617){} (16) -- node(e1516){} (15) -- node(e1415){} (14) -- node(e1014){} (10) -- node(e910){} (9) -- node(e59){} (5) -- node(e15){} (1) -- cycle;
  \draw (5) -- node(e56){} (6) -- node(e611){} (11) -- node(e1112){} (12) -- node(e1217){} (17);
  \draw (3) -- node(e37){} (7) -- node(e712){} (12) -- node(e1213){} (13);
  \draw (6) -- node(e67){} (7) -- node(e78){} (8);
  \draw (2) -- node(e26){} (6) -- node(e610){} (10);
  \draw (11) -- node(e1114){} (14);
  \draw (15) -- node(e1215){} (12);
  \draw [dashed,gray=10!white] plot [smooth] coordinates {(e59) (e610) (e1114) (e1215) (e1617)};
  \draw [dashed,gray=10!white] plot [smooth] coordinates {(e15) (e26) (e37) (e48)};
  \draw [dashed,gray=10!white] plot [smooth] coordinates {(e23) (e67) (e1112) (e1415)};
  \draw [dashed,gray=10!white] plot [smooth] coordinates {(e34) (e78) (e1213) (e1718)};
  \draw [dashed,gray=10!white] plot [smooth] coordinates {(e12) (e56) (e910)};
  \draw [dashed,gray=10!white] plot [smooth] coordinates {(e1014) (e611) (e712) (e813)};
  \draw [dashed,gray=10!white] plot [smooth] coordinates {(e1516) (e1217) (e1318)};
  \draw (0.5,7) node{(a)};
\end{tikzpicture}
\hspace{30pt}
\begin{tikzpicture}[thick,scale=0.75]
  \draw[transparent] (1,5) -- node(a){} (1,7) -- node(b){} (3,7) -- node(c){} (3,5) -- node(d){} (1,5) -- (1,7);
  \draw[dashed,gray=10!white, shorten <=-15, shorten >=-15] (a) -- (c) node[right=10]{$\alpha$};
  \draw[dashed,gray=10!white, shorten <=-15, shorten >=-15] (b) -- (d) node[below=10]{$\beta$};
  \draw (1,5) -- node[fill=white,inner sep=3](a){$u$} (1,7) -- node(b)[fill=white,inner sep=3]{$\wt{v}$} (3,7) -- node(c)[fill=white,inner sep=3]{$\wh{u}$} (3,5) -- node(d)[fill=white,inner sep=3]{$v$} (1,5) -- (1,7);
  \draw (0,7) node{(b)};
  \draw (1,1) node[fill=white]{$x$} -- node(aa){} (1,3) node[fill=white]{$\wt{x}$} -- node(bb){} (3,3) node[fill=white]{$\th{x}$} -- node(cc){} (3,1) node[fill=white]{$\wh{x}$} -- node(dd){} (1,1) -- (1,3);
  \draw[dashed,gray=10!white, shorten <=-15, shorten >=-15] (aa) -- (cc) node[right=10]{$p$};
  \draw[dashed,gray=10!white, shorten <=-15, shorten >=-15] (bb) -- (dd) node[below=10]{$q$};
  \draw (0,3) node{(c)};
\end{tikzpicture}
\end{center}
\caption{(a) Planar quad-graph, (b) quad with variables on edges, (c) quad with variables on vertices. Parameters are associated with the dashed lines, which represent characteristics.}
\label{quadgraph}
\end{figure}
A class of transformations connecting systems with variables on edges, such as (\ref{F1}), with polynomial quad equations in which variables are on vertices, were introduced by Papageorgiou, Tongas and Veselov \cite{PTV} and developed by Kassotakis and Nieszporski \cite{KaNie}.
It is a new transformation in this class which is the point of departure in this article, it involves equations on edges of the graph which around a single quad take the form
\begin{equation}
\begin{split}
&B(u,x,\wt{x},\alpha)=0, \quad B(v,x,\wh{x},\beta)=0,\\
&B(\wh{u},\wh{x},\th{x},\alpha)=0, \quad B(\wt{v},\wt{x},\th{x},\beta)=0,
\end{split}\label{BT}
\end{equation}
where $B$ is the following polynomial
\begin{multline}
B(u,x,\wt{x},\alpha):=\\ \ (1+cx)(1+c\wt{x})u^2-[(1+c^2)(1+x\wt{x})\alpha+2c(x+\wt{x})]u+(c+x)(c+\wt{x})\alpha,
\end{multline}
and $c\in\mathbb{C}\setminus\{0,1,-1,i,-i\}$ is a constant parameter.
This system connects model (\ref{F1}) with the primary quad equation found in \cite{AtkNie}:
\begin{equation}
\begin{split}
&(p-q)[(c^{-2}p-c^2q)(x\wt{x}-\wh{x}\th{x})^2-(c^{-2}q-c^2p)(x\wh{x}-\wt{x}\th{x})^2]\\
&\quad -(p-q)^2[(x+\th{x})^2(1+\wt{x}^2\wh{x}^2)+(\wt{x}+\wh{x})^2(1+x^2\th{x}^2)]\\
&\quad +[(x-\th{x})(\wt{x}-\wh{x})(c^{-1}-cpq)-2(p-q)(1+x\wt{x}\wh{x}\th{x})]\\
&\quad \times [(x-\th{x})(\wt{x}-\wh{x})(c^{-1}pq-c)-2(p-q)(x\th{x}+\wt{x}\wh{x})]=0,
\end{split}
\label{QQ4}
\end{equation}
where
\begin{equation}
p=\mm(\alpha), \quad q=\mm(\beta), \quad \mm:=\alpha\mapsto\frac{2c^2-(1+c^2)\alpha}{c(1+c^2)\alpha-2c},\label{mm}
\end{equation}
see Figure \ref{quadgraph}(c).
The quad equation (\ref{QQ4}), denoted $Q4^*$, is quadratic in the dependent variable, but the basis of its original construction guarantees existence of a rational reformulation.
In fact the system on edges (\ref{BT}) is a way to introduce auxiliary edge variables, the (autonomous) rational reformulation in terms of these variables being exactly system (\ref{F1}).
Typically such reformulation would involve both edge and vertex variables, so relations (\ref{BT}) are especially natural because in the resulting reformulation, dependence on the original (vertex) variables disappears.
It means system (\ref{BT}) is also a Miura-type transformation obtaining solutions of (\ref{QQ4}) from solutions of (\ref{F1}).
This can be verified directly, specifically a calculation will verify the following two non-trivial facts.
\begin{prop}
(i) Elimination of variables $\th{x},\wt{x},\wh{x}$ from system (\ref{BT}) yields a quadratic polynomial equation in $x$ whose coefficients vanish exactly when (\ref{F1}) holds.
(ii) Elimination of $u,\wh{u},v,\wt{v}$ between (\ref{F1}) and (\ref{BT}) yields a single equation in the remaining variables, which is exactly (\ref{QQ4}).
\end{prop}
Therefore, assuming that $u,\wh{u},v,\wt{v}$ are such that (\ref{F1}) holds, and choosing $x$ freely, system (\ref{BT}) uniquely determines $\wt{x}$, $\wh{x}$ and $\th{x}$, which then satisfy (\ref{QQ4}).
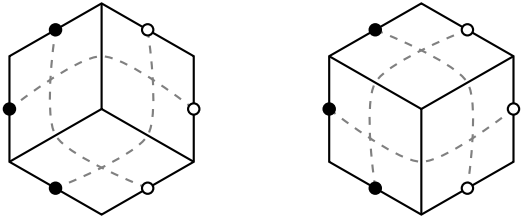
\begin{figure}[t]
\begin{center}
\begin{tikzpicture}[thick,scale=0.7]
  \tikzstyle{every node}=[inner sep=0pt]
  \coordinate (1) at (30:2);
  \coordinate (2) at (90:2);
  \coordinate (3) at (150:2);
  \coordinate (4) at (210:2);
  \coordinate (5) at (270:2);
  \coordinate (6) at (330:2);
  \coordinate (7) at (0,0);
  \draw[transparent] (1) -- node(e12){} (2) -- node(e23){} (3) -- node(e34){} (4) -- node(e45){} (5) -- node(e56){} (6) -- node(e16){} (1) -- cycle;
  \draw[transparent] (2) -- node(e27){} (7) -- node(e67){} (6);
  \draw[transparent] (4) -- node(e47){} (7);
  \draw [dashed,gray=10!white] plot [smooth] coordinates {(e16) (e27) (e34)};
  \draw [dashed,gray=10!white] plot [smooth] coordinates {(e12) (e67) (e45)};
  \draw [dashed,gray=10!white] plot [smooth] coordinates {(e56) (e47) (e23)};
  \draw (1) -- node[draw,circle,fill=white,inner sep=1.5]{} (2) -- node[draw,circle,fill=black,inner sep=1.5]{} (3) -- node[draw,circle,fill=black,inner sep=1.5]{} (4) -- node[draw,circle,fill=black,inner sep=1.5]{} (5) -- node[draw,circle,fill=white,inner sep=1.5]{} (6) -- node[draw,circle,fill=white,inner sep=1.5]{} (1) -- cycle;
  \draw (2) -- node{} (7) -- node{} (6);
  \draw (4) -- node{} (7);
\end{tikzpicture}
\hspace{40pt}
\begin{tikzpicture}[thick,scale=0.7]
  \tikzstyle{every node}=[inner sep=0pt]
  \coordinate (4) at (30:2);
  \coordinate (5) at (90:2);
  \coordinate (6) at (150:2);
  \coordinate (1) at (210:2);
  \coordinate (2) at (270:2);
  \coordinate (3) at (330:2);
  \coordinate (7) at (0,0);
  \draw[transparent] (1) -- node(e12){} (2) -- node(e23){} (3) -- node(e34){} (4) -- node(e45){} (5) -- node(e56){} (6) -- node(e16){} (1) -- cycle;
  \draw[transparent] (2) -- node(e27){} (7) -- node(e67){} (6);
  \draw[transparent] (4) -- node(e47){} (7);
  \draw [dashed,gray=10!white] plot [smooth] coordinates {(e16) (e27) (e34)};
  \draw [dashed,gray=10!white] plot [smooth] coordinates {(e12) (e67) (e45)};
  \draw [dashed,gray=10!white] plot [smooth] coordinates {(e56) (e47) (e23)};
  \draw (1) -- node[draw,circle,fill=black,inner sep=1.5]{} (2) -- node[draw,circle,fill=white,inner sep=1.5]{} (3) -- node[draw,circle,fill=white,inner sep=1.5]{} (4) -- node[draw,circle,fill=white,inner sep=1.5]{} (5) -- node[draw,circle,fill=black,inner sep=1.5]{} (6) -- node[draw,circle,fill=black,inner sep=1.5]{} (1) -- cycle;
  \draw (2) -- node{} (7) -- node{} (6);
  \draw (4) -- node{} (7);
\end{tikzpicture}
\end{center}
\caption{The Yang-Baxter property, global transfer of data along characteristics is independent of the order in which they cross.
This is equivalent to the local condition that equations for variables on white nodes in terms of variables on black ones are the same for both of these quad graphs.}
\label{YB}
\end{figure}

In fact these features can be explained in terms of the integrability of (\ref{F1}), namely the Yang-Baxter property, or cubic consistency, see Figure \ref{YB}.
The main observation is that
\begin{equation}
\begin{split}
& B(u,\mm(w),\mm(\wt{w}),\alpha)=0 \quad \Leftrightarrow \quad \\
& \wt{w} = \gamma u \frac{(1-\alpha)w+\alpha-\gamma-(1-\gamma)u}{(1-\gamma)\alpha w+(\gamma-\alpha)wu-(1-\alpha)\gamma u}, \quad \gamma=\left(\frac{2c}{1+c^2}\right)^2,
\end{split}\label{BF1}
\end{equation}
thus (\ref{BT}) is really nothing but the same underlying model (\ref{F1}).
The system of all equations (\ref{F1}), (\ref{BT}), (\ref{QQ4}) is associated with the non-planar cube that, due to the consistency property, may be formed as the union of the two graphs in Figure \ref{YB} without breaking the well-posedness of the indicated initial value problem.
This cube is shown in projected form in Figure \ref{cube} with variables and parameters added explicitly.
It can be verified by calculation that as a consequence of imposing (\ref{F1}) on each face, the variables on the four edges through which the closed $\gamma$ characteristic passes are related by (\ref{QQ4}) via the substitution
\begin{equation}
\begin{split}
&x=\mm(w),\ \wt{x}=\mm(\wt{w}),\ \wh{x}=\mm(\wh{w}),\ \th{x}=\mm(\th{w}), \\ & p=\mm(\alpha), \ q=\mm(\beta), \ c=\mm(\gamma).
\end{split}\label{sbs}
\end{equation}
This M\"obius change of variables brings the resulting multi-quadratic quad-equation to the canonical form that was given originally in \cite{AtkNie} (\ref{QQ4}).
The generic form of (\ref{QQ4}) will be given later in this article.
The M\"obius transformation $\mm$ here and in (\ref{mm}) is characterised by its action $\mm:(0,1,\infty)\mapsto(-c,1/c,-1/c)$.

Up to the same change of variables, the equations (\ref{BT}) are associated with the four quads through which the $\gamma$ characteristic passes.
\begin{figure}[t]
\begin{center}
\begin{tikzpicture}[thick,scale=1.3]
  \draw [dashed,gray=10!white] plot [smooth cycle] coordinates {(45:1.68) (135:1.68) (225:1.68) (315:1.68)};
  \draw [dashed,gray=10!white] (0:1.75) -- (180:1.75);
  \draw [dashed,gray=10!white] (90:1.75) -- (270:1.75);
  \tikzstyle{every node}=[circle,fill=white,inner sep=0.7pt]
  \draw (45:1) -- node{$\wt{v}$} (135:1) -- node{$u$} (225:1) -- node{$v$} (315:1) -- node{$\wh{u}$} (45:1) -- cycle;
  \draw (45:2.4) -- node{$\bt{v}$} (135:2.4) -- node{$\wb{u}$} (225:2.4) -- node{$\wb{v}$} (315:2.4) -- node{$\hb{u}$} (45:2.4) -- cycle;
  \draw (45:1) -- node{$\th{w}$} (45:2.4);
  \draw (135:1) -- node{$\wt{w}$} (135:2.4);
  \draw (225:1) -- node{$w$} (225:2.4);
  \draw (315:1) -- node{$\wh{w}$} (315:2.4);
  \draw[gray=10!white] (115:1.52) node{$\gamma$};
  \draw[gray=10!white] (180:1.15) node{$\alpha$};
  \draw[gray=10!white] (90:1.15) node{$\beta$};
\end{tikzpicture}
\end{center}
\caption{A quad-graph domain for (\ref{F1}). The four variables along the closed characteristic $\gamma$ are related by the multi-quadratic quad equation (\ref{QQ4}) modulo the change of variables (\ref{sbs}).}
\label{cube}
\end{figure}
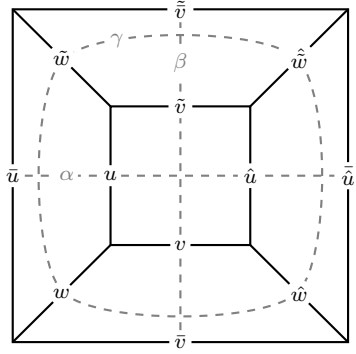
Upon this observation the statement that (\ref{BT}) constitutes a Miura-type transformation between (\ref{F1}) and (\ref{QQ4}) can be identified with the statement of the Yang-Baxter property of (\ref{F1}) itself.

The theory of the simpler rational model (\ref{F1}) therefore encompasses the transformation (\ref{BT}) and the multi-quadratic model (\ref{QQ4}).
In hindsight one can notice that the basis of construction of (\ref{F1}) and (\ref{QQ4}) are in fact similar. 
In the first instance it is the preservation of rationality under permutation of the variables of the model, and in the second it is existence of a rational reformulation.

However, there is one particular element from the theory of (\ref{QQ4}) which turns out to give new insight into (\ref{F1}).
It is known \cite{AtkNie} that such polynomials as $B$ that define a B\"acklund transformation or rational re-formulation for (\ref{QQ4}) have the following discriminant property,
\begin{multline}
\dis[B(u,\mm(w),\mm(\wt{w}),\alpha),u] = 0 \quad \Leftrightarrow \quad \\
\left[{(1-\gamma)^{-1}}(w-\gamma)(\wt{w}-\gamma)(\alpha-\gamma) + \gamma(w+\wt{w}+\alpha-\gamma)\right]^2 - 4\gamma w\wt{w}\alpha = 0,\label{stq}
\end{multline}
where the polynomial appearing in (\ref{stq}) is the Weierstrass-type symmetric triquadratic with discriminant polynomial $r(w) = w(w-1)(w-\gamma)$.
This characterises $B$ as being amongst a particular class of polynomials that we call idempotent, as shown in Appendix A.
Due to (\ref{BF1}), the model (\ref{F1}) is also connected with this class of polynomials, and it is this connection which is the key observation we develop in this paper.
The idempotent class has some remarkable features of independent interest, and the connection with (\ref{F1}) can be established with greater directness after some of those features are understood.

\section{Idempotent biquadratic correspondences}\label{IB}
We define and give basic properties of the relevant class of algebraic correspondences.

Symmetric biqudaratic polynomials,
\begin{equation}
h(x,y) = c_0+c_1(x+y)+c_2xy+c_3(x^2+y^2)+c_4xy(x+y)+c_5x^2y^2, \label{sb}
\end{equation}
$c_0,\ldots,c_5\in\mathbb{C}$, provided they don't decompose into a product involving one-variable polynomial factors, define dynamics on {\it orbits}, that is sequences of values 
\begin{equation}
\ldots,x_{-2},x_{-1},x_0,x_1,x_2,\ldots\in\mathbb{C}\cup\{\infty\}\label{orbit}
\end{equation}
such that 
\begin{equation}
h(x_n,y)=0\ \Leftrightarrow \ y\in\{x_{n-1},x_{n+1}\}, \quad n\in\mathbb{Z}.\label{orbdef}
\end{equation}
The two-valued dynamics defined by the biquadratic manifests only in the choice of direction along the orbit.
The significance of the orbits, which are sequences defined uniquely up to orientation, is that they may be taken as a definition of elliptic functions in the discrete setting.

The self-composition of the correspondence defined by $h$ yields another symmetric biquadratic polynomial, $g$, via the formula
\begin{equation}
\res_2[h(x,y),h(y,z),y] = (x-z)^2g(x,z),\label{scomp}
\end{equation}
where $\res_2$ denotes the resultant of quadratic polynomials\footnote{More precisely $\res_2[p(x),q(x),x]$ is the determinant of the Sylvester matrix of $p$ and $q$ with the explicit assumption that both $p$ and $q$ are degree two polynomials, allowing a possibly zero coefficient of the second-degree monomial term.}.
The other factor appearing in (\ref{scomp}), $(x-z)^2$, is due to the fact that $h$ is symmetric.
This notion of self-composition allows to define the idempotent class.
\begin{definition}\label{idempotent}
We refer to $h$ (\ref{sb}) as idempotent if its self-composition $g$, defined by (\ref{scomp}), is a non-zero scalar multiple of $h$.
\end{definition}
Note that if $h$ has a one-variable polynomial factor, or it drops in degree, then its self-composition vanishes identically.
Such degenerate cases are excluded in Definition \ref{idempotent} by requiring the scalar multiple be non-zero.
In particular, idempotent biquadratic polynomials have well-defined orbits.

The utility to be found in biquadratic polynomials from this class begins with the following list of complementary ways to characterise them.
\begin{prop}\label{character}
Assume the symmetric biquadratic polynomial $h$ in (\ref{sb}) is not a scalar multiple of $(x-y)^2$, then the following are equivalent,
\begin{itemize}
\item[(i)] $h$ is idempotent,
\item[(ii)$^*$] all orbits of $h$ are three-periodic ($x_{n+3}=x_n$ in (\ref{orbit})),
\item[(iii)$^*$] there exists a non-constant three-periodic orbit of $h$,
\item[(iv)$^*$] the coefficients of $h$ satisfy the constraint $c_1c_4+c_3^2=c_0c_5+c_2c_3$,
\item[(v)$^*$] $h$ satisfies the four-variable polynomial identity
\begin{multline}
(w-z)h(w,z)(x-y)h(x,y)
+(w-y)h(w,y)(z-x)h(z,x)\\
+(w-x)h(w,x)(y-z)h(y,z)=0,\label{fe}
\end{multline}
\item[(vi)] $h$ can be written in the form
\begin{equation}
h(x,y) = \frac{r_1(x)r_2(y)-r_2(x)r_1(y)}{x-y}, \label{redrag}
\end{equation}
where $r_1,r_2$ are cubic polynomials with no common roots.
\end{itemize}
Here the $^*$ indicates addition of the non-degeneracy condition on $h$ that its self-composition, $g$ in (\ref{scomp}), does not vanish identically.
\end{prop}
\begin{proof}
Deleting every other point from an orbit of $h$ yields an orbit of its self-composition $g$ defined by (\ref{scomp}), considering this and the fact that orbits of $g$ and $h$ coincide due to (i), leads to (ii)$^*$.

Existence of a non-constant orbit is a consequence of the assumption that $h$ is not a scalar multiple of $(x-y)^2$, therefore (ii)$^*$ implies (iii)$^*$.

To see that (iii)$^*$ implies (iv)$^*$ denote the elements of the three-periodic orbit mentioned in (iii)$^*$ by $x_0$, $x_1$ and $x_2$, at least two of which are distinct, and consider the resulting three equalities between polynomials in $y$:
\begin{equation}
\begin{split}
&{h(x_0,y)}=({c_3+c_4x_0+c_5x_0^2})(y-x_1)(y-x_2),\\
&{h(x_1,y)}=({c_3+c_4x_1+c_5x_1^2})(y-x_2)(y-x_0),\\
&{h(x_2,y)}=({c_3+c_4x_2+c_5x_2^2})(y-x_0)(y-x_1).
\end{split}
\end{equation}
Substituting the explicit form of $h$ (\ref{sb}) followed by elimination of $x_0$, $x_1$, $x_2$ yields the condition on the coefficients which appears in (iv)$^*$.

Substitution of (\ref{sb}) directly shows that polynomial identity (\ref{fe}) is satisfied as a consequence of the condition on the coefficients appearing in (iv)$^*$, so (iv)$^*$ implies (v)$^*$.

Choosing $w$ and $z$ to be fixed constants such that $\eta^2:=(w-z)h(w,z)\neq 0$ it is clear that
\begin{equation}
\frac{1}{\eta}(w-x)h(w,x) = r_1(x), \quad \frac{1}{\eta}(z-x)h(z,x) = r_2(x),\label{r12}
\end{equation}
are cubic polynomials in $x$, and writing (\ref{fe}) in terms of these one obtains (\ref{redrag}).
Furthermore $r_1,r_2$ have no common roots due to the non-degeneracy assumption on $h$, therefore (v)$^*$ implies (vi).

Assuming $h$ is of the form (\ref{redrag}) specified in (vi), a direct calculation shows that
\begin{equation}
\res_2[h(x,y),h(y,z),y]=(x-z)^2\mu h(x,z), \quad \mu:=\res_3[r_1(y),r_2(y),y].\label{ec}
\end{equation}
The scalar multiple $\mu$ appearing in (\ref{ec}) is non-zero if and only if $r_1,r_2$ in (\ref{redrag}) have no common roots (res$_3$ excluding also the possibility of $\infty$ being a common root, which corresponds to a drop in degree of both polynomials), therefore (vi) implies (i).
\end{proof}
Equivalences (ii)$^*$ and (iii)$^*$ connect the idempotent class of biquadratics with the trisection of the elliptic function periods.
The three-periodic orbits will play the central role in what follows.
They are specified completely by the un-ordered triplet of participating variables, orbits being determined only up to orientation.
Therefore within this class of biquadratics the un-ordered triplet itself may be referred to as the orbit, and, to reflect the symmetry, it is natural think of variables in an orbit as being arranged in a triangle as in Figure \ref{tri}.
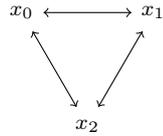
\begin{figure}[t]
\begin{center}
\begin{tikzpicture}
  \tikzstyle{every node}=[circle, fill=white, inner sep=0pt, minimum width=16pt]
  \node (a) at (150:1) {$x_0$};
  \node (b) at (270:1) {$x_2$} edge [<->] (a);
  \node (c) at (30:1) {$x_1$} edge [<->] (a) edge [<->] (b);
\end{tikzpicture}
\end{center}
\caption{An orbit of an idempotent biquadratic: $h(x_i,y)=0 \Leftrightarrow y\in\{x_j,x_k\}$, $\{i,j,k\} = \{0,1,2\}$.}
\label{tri}
\end{figure}

We remark that the idempotents occupy a similar special position in relation to the biquadratic correspondences as do M\"obius involutions to the group of M\"obius transformations.
This was observed in \cite{atk3} based on the similarity of characteristic features described by Proposition \ref{character} to characteristics of the M\"obius involutions, but can also be seen directly by considering formula (\ref{redrag}) in the degenerate case when $r_1$ and $r_2$ are degree two polynomials in which the correspondence defined by $h$ becomes a M\"obius involution.

The formula (\ref{redrag}) is well-known.
It was observed by Caley that this formula encodes Bezout's characterisation of the resultant of polynomials $r_1$ and $r_2$ (see \cite{hkw}), specifically the resultant is the determinant of the coefficient matrix of (\ref{redrag}).
The appearance of this formula here is interesting, but perhaps not completely surprising given the role played by the resultant in Definition \ref{idempotent}.

As a final remark, we note that the pattern of terms and variables in (\ref{fe}) is very similar to that of the three-term sigma-function identity (see p. 390 of \cite{han}), or Riemann relation, though an explanation for this is not so clear.

The characterising polynomial identity (\ref{fe}) leads directly to the following constructive existence-and-uniqueness result for the idempotent biquadratics.
\begin{lemma}\label{exists}
Given a pair of disjoint triplets taken from $\mathbb{C}\cup\{\infty\}$
\begin{equation}
x_0,x_1,x_2, \quad y_0,y_1,y_2,\label{triplets}
\end{equation}
there exists an idempotent biquadratic, which is unique up to a scalar multiple, for which these are orbits.
It may be expressed via the formula (\ref{redrag}) by choosing
\begin{equation}
r_1(x)=(x-x_0)(x-x_1)(x-x_2), \quad r_2(x)=(x-y_0)(x-y_1)(x-y_2).\label{r12roots}
\end{equation}
\end{lemma}
\begin{proof}
For existence it is sufficient to check that the roots of the polynomials $r_1$ and $r_2$ give two particular orbits of $h$ in (\ref{redrag}).
Conversely, supposing (\ref{triplets}) are the orbits of a given idempotent biquadratic polynomial $h$, choose $w=x_0$ and $z=y_0$ in the identity characterising such polynomials (\ref{fe}). 
By inspection it can be seen that the polynomials $r_1$, $r_2$ appearing in (\ref{redrag}), (\ref{r12}), which is just a re-arrangement of (\ref{fe}), vanish on the orbits (\ref{triplets}) due to the choice of $w$ and $z$.
In other words, any such biquadratic $h$ coincides with (\ref{redrag}), (\ref{r12roots}) up to a scalar multiple.
\end{proof}
Thus two orbits can be chosen freely, which then determine an idempotent biquadratic polynomial uniquely up to a scalar multiple.
Clearly three orbits cannot be chosen freely.

\section{Rationality of the three-orbit constraint}\label{IB2}
We connect the previously defined class of algebraic correspondences with the scalar quadrirational dynamical systems, this emerges directly by looking more closely at the three-orbit constraint.
\subsection{The three-orbit constraint}

The following form of this constraint is a straightforward consequence of Lemma \ref{exists}.
\begin{prop}\label{constraint}
Denote three elements in the vector space of cubic polynomials by $r_1$, $r_2$ and $r_3$.
The roots of these three polynomials are three distinct orbits of an idempotent biquadratic correspondence if and only if there are no pair-wise common roots and they are linearly dependent:
\begin{equation}
\exists \gamma_1,\gamma_2,\gamma_3 \in\mathbb{C}\setminus \{0\} \ {\textrm{ such that }}\  \gamma_1r_1(x)+\gamma_2r_2(x)+\gamma_3r_3(x)=0 \quad \forall x.\label{nsc}
\end{equation}
\end{prop}
An important feature of the three-orbit constraint is its permutation symmetry, by which we mean the group of permutations of the variables under which the constraint is invariant.
The full permutation symmetry group, which is clear from Proposition \ref{constraint}, is obscured when the constraint is written in the following more explicit form, which nevertheless complements Proposition \ref{constraint} by exhibiting its rationality.
\begin{prop}\label{constraint2}
Three disjoint triplets taken from $\mathbb{C}\cup\{\infty\}$,
\begin{equation}
x_0,x_1,x_2, \quad y_0,y_1,y_2, \quad z_0,z_1,z_2,\label{tt}
\end{equation}
are orbits of an idempotent biquadratic polynomial if and only if they satisfy the conditions
\begin{equation}
h(x_0,x_1;z_0,z_1,z_2;y_0,y_1,y_2)=0, \quad h(y_0,y_1;z_0,z_1,z_2;x_0,x_1,x_2)=0, \label{ratmap}
\end{equation}
where $h$ is the polynomial defined by expression
\begin{multline}
h(x,y;x_0,x_1,x_2;y_0,y_1,y_2):=\\
(x-y)^{-1}[(x-x_0)(x-x_1)(x-x_2)(y-y_0)(y-y_1)(y-y_2)-\\(y-x_0)(y-x_1)(y-x_2)(x-y_0)(x-y_1)(x-y_2)],\label{genh}
\end{multline}
from (\ref{redrag}), (\ref{r12roots}), cf. Lemma \ref{exists}. 
\end{prop}
\begin{proof}
According to Lemma \ref{exists}, the equation on the left in (\ref{ratmap}) expresses that $x_0$ and $x_1$ belong to an orbit of the idempotent biquadratic polynomial determined by the pair of orbits $y_0,y_1,y_2$ and $z_0,z_1,z_2$.
Similarly, the equation on the right in (\ref{ratmap}) expresses that $y_0$ and $y_1$ belong to an orbit of the idempotent biquadratic polynomial determined by the pair of orbits $x_0,x_1,x_2$ and $z_0,z_1,z_2$.
The conditions (\ref{ratmap}) are therefore necessary for the claimed existence, and it remains to demonstrate that they are also sufficient.

According to Proposition \ref{constraint} it is sufficient that conditions (\ref{ratmap}) lead to the linear dependence of the three polynomials
\begin{equation}
\begin{split}
r_1(x)&=(x-x_0)(x-x_1)(x-x_2),\\
r_2(x)&=(x-y_0)(x-y_1)(x-y_2),\\
r_3(x)&=(x-z_0)(x-z_1)(x-z_2).
\end{split}
\end{equation}
This linear dependence can be verified by calculation as follows. Choose
\begin{equation}
\gamma_3=-1, \quad \gamma_2=\frac{r_3(x_1)}{r_2(x_1)}, \quad \gamma_1=\frac{r_3(y_1)}{r_1(y_1)},\label{gamchoice}
\end{equation}
in the polynomial appearing in (\ref{nsc}).
This polynomial then vanishes identically upon substitution of $x_2$ and $y_2$ determined by (\ref{ratmap}). 
\end{proof}

It is clear from its general form (\ref{ratmap}) that amongst the participating variables (\ref{tt}), the three-orbit constraint determines $y_2$ and $x_2$ rationally from the remaining seven variables.
Combining this rationality with the particular permutation symmetries $x_1\leftrightarrow x_2$ and $y_1\leftrightarrow y_2$ (which are clear from Proposition \ref{constraint}), we have the following.
\begin{corollary}[following from Propositions \ref{constraint} and \ref{constraint2}]\label{qrr}
With $x_0$, $y_0$, $z_0$, $z_1$ and $z_2$ fixed, the three-orbit constraint (\ref{ratmap}), visualised by assigning the remaining four variables $x_1$, $x_2$, $y_1$ and $y_2$ to edges of a quad as in Figure \ref{quad}, determines any pair of adjacent variables rationally from those on the opposing edges.
\end{corollary}
\begin{figure}[t]
\begin{center}
\begin{tikzpicture}[thick,scale=0.9]
  \draw (0,0) -- node[circle,fill=white,inner sep=1pt]{$x_2$} (1.5,1.5) -- node[circle,fill=white,inner sep=1pt]{$y_2$} (3,0) -- node[circle,fill=white,inner sep=1pt]{$x_1$} (1.5,-1.5) -- node[circle,fill=white,inner sep=1pt]{$y_1$} (0,0) -- cycle;
\end{tikzpicture}
\end{center}
\caption{Quadrirationality: A system of equations in four variables assigned to edges of a quad, with the feature that variables on any pair of adjacent edges are determined rationally from the variables on the opposing edges.}
\label{quad}
\end{figure}
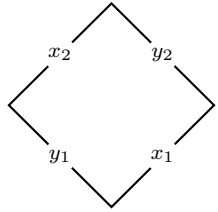
By definition, this system therefore falls within the {\it quadrirational} class introduced in \cite{ABSf}.

\subsection{Equivalence with the quadrirational Yang-Baxter maps}
Up to a M\"obius change of variables, there are five non-degenerate scalar quadrirational systems \cite{ABSf}, which have been designated $F_{I},\ldots,F_{V}$.
To make contact with the canonical forms of those systems, we proceed first by using the action of the M\"obius group to reduce the 5-parameter quadrirational system described in Corollary \ref{qrr} to a two-parameter family.
Such transformations can bring the variables $z_0,z_1,z_2$ of the third orbit to some fixed canonical values, but will not alter the size of the set $\{z_0,z_1,z_2\}$, so this procedure leads to three distinct cases depending on this size.
\begin{prop}\label{conform}
The constraint on three orbits of an idempotent biquadratic polynomial (\ref{ratmap}), in the case where variables of the third orbit, $(z_0,z_1,z_2)$, are set equal to $(0,1,\infty)$, takes the form
\begin{equation}
(1-x_0)(1-x_1)(1-x_2) = (1-y_0)(1-y_1)(1-y_2),\quad x_0x_1x_2=y_0y_1y_2,\label{F1a}
\end{equation}
in the case $(z_0,z_1,z_2)=(0,\infty,\infty)$ it takes the form
\begin{equation}
x_0x_1x_2=y_0y_1y_2, \quad x_0+x_1+x_2=y_0+y_1+y_2,\label{F2a}
\end{equation}
and in the case $(z_0,z_1,z_2)=(\infty,\infty,\infty)$ it takes the form
\begin{equation}
x_0+x_1+x_2=y_0+y_1+y_2,\quad x_0^2+x_1^2+x_2^2=y_0^2+y_1^2+y_2^2.\label{F4a}
\end{equation}
\end{prop}
\begin{proof}
This is by calculation.
The canonical forms (\ref{F1a}), (\ref{F2a}) and (\ref{F4a}) can be compared with (\ref{ratmap}) when $(z_0,z_1,z_2)$ are taken to be the respective canonical values.
To make the comparison it is useful to re-arrange the above expressions to give $x_2$ and $y_2$ as rational expressions in the remaining variables.
\end{proof}

The second step, completing the connection from the quadrirational system described in Corollary \ref{qrr} to the canonical forms of \cite{ABSf} (in the notation of Section \ref{F1Q4}), involves just making some substitutions and rearrangements.
\begin{prop}[Connection with canonical forms in \cite{ABSf}]\label{YBid}
On substitution of $\{x_0,x_1,x_2\}=\{\beta,u,\wh{u}\}$ and  $\{y_0,y_1,y_2\}=\{\alpha,v,\wt{v}\}$ and subsequent re-arrangement, the system (\ref{F1a}) is equivalent to $F_I$:
\begin{equation}
\begin{split}
\wh{u} &= \alpha v \frac{(1-\beta)u+\beta-\alpha-(1-\alpha)v}{(1-\alpha)\beta u+(\alpha-\beta)uv-(1-\beta)\alpha v},\\
\wt{v} &= \beta u \frac{(1-\alpha)v+\alpha-\beta-(1-\beta)u}{(1-\beta)\alpha v+(\beta-\alpha)vu-(1-\alpha)\beta u}.
\end{split}\label{listF1}
\end{equation}
On substitution of $\{x_0,x_1,x_2\}=\{\beta,\alpha u,\alpha \wh{u}\}$ and $\{y_0,y_1,y_2\}=\{\alpha,\beta v,\beta \wt{v}\}$, system (\ref{F2a}) is equivalent to $F_{II}$:
\begin{equation}
\begin{split}
\wh{u} &= \frac{v(\alpha u-\beta v+\beta-\alpha)}{\beta(u-v)},\\
\wt{v} &= \frac{u(\alpha u-\beta v+\beta-\alpha)}{\alpha(u-v)}.
\end{split}\label{listF2}
\end{equation}
On substitution of $\{x_0,x_1,x_2\}=\{\beta,\alpha-u,\alpha-\wh{u}\}$ and $\{y_0,y_1,y_2\}=\{\alpha,\beta-v,\beta-\wt{v}\}$, system (\ref{F4a}) is equivalent to $F_{IV}$:
\begin{equation}
\begin{split}
\wh{u} &= v \left( 1 + \frac{\beta-\alpha}{u-v}\right),\\
\wt{v} &= u \left( 1 + \frac{\beta-\alpha}{u-v}\right).
\end{split}\label{listF4}
\end{equation}
\end{prop}

The dynamics defined by the maps $F_I$, $F_{II}$ and $F_{IV}$ are therefore associated with a family of idempotent biquadratic correspondences that have a common fixed orbit, the distinction between the three maps being the number of distinct values in that orbit.
Besides this algebraic interpretation, the most salient new feature of the maps which is evident on the level of the three-orbit constraint is the previously hidden permutation symmetry group, which extends the symmetries of the quad (Figure \ref{quadgraph}(b)) replacing it with a triplet-pair.
The remaining quadrirational maps from \cite{ABSf}, namely $F_{III}$ and $F_{V}$, are obtained by limiting procedures that erode the connection to the idempotent biquadratics and break the additional permutation symmetry.
This means those maps do not arise naturally in the present context, and the subsequent theory developed here will not apply directly to them.

\subsection{Yang-Baxter maps and the hypercube}\label{NCUBE}
It is a combinatorial fact that the planar quad-graph with $n$ labelled characteristics (cf. Section \ref{F1Q4}), can be embedded in an $n$-dimensional hypercube (or $n$-cube), provided none of the characteristics self-intersect \cite{AdVeQ}.
The $n$-cube is a non-planar quad-graph domain.
For systems (like those listed in Proposition \ref{YBid}) that have the symmetry of the quad, the Yang-Baxter property is exactly the condition that a solution on the quad-graph can be extended to a solution on the embedding hypercube; the hypercube is therefore the more general, and also a more regular lattice geometry to associate with these systems.
Because the hypercube lattice geometry will play a key role subsequently in this paper, we formulate the following proposition giving coordinatisation and initial value problem.
The proposition is stated in terms of the idempotent-biquadratic three-orbit constraint, but is a standard consequence of the Yang-Baxter property of the equivalent maps listed in Proposition \ref{YBid}.
\begin{prop}\label{YBP}
Denote $n+n2^{n-1}$ variables by
\begin{equation}
y_i, \quad x_i^I, \qquad i \in \{1,\ldots,n\}, \quad I\subseteq \{1,\ldots,n\} \setminus \{i\}.
\end{equation}
Choose any $z_0,z_1,z_2\in\mathbb{C}\cup\{\infty\}$, and impose the idempotent-biquadratic three-orbit constraint (cf. Propositions \ref{constraint2} and \ref{conform}) on all sets of variables
\begin{equation}
\qquad y_i,x_j^I,x_j^{I\cup\{i\}}, \quad y_j,x_i^I,x_i^{I\cup\{j\}}, \quad z_0,z_1,z_2, \label{cc}
\end{equation}
for
\begin{equation}
i,j\in\{1,\ldots,n\}, \quad i\neq j, \quad I\subseteq \{1,\ldots,n\} \setminus \{i,j\}.
\end{equation}
Then variables remaining from the whole set are rational functions of the unconstrained subset 
\begin{equation}
x_1^{\{\}},\ldots,x_n^{\{\}},y_1,\ldots,y_n.\label{cubeivp}
\end{equation}
\end{prop}
Here the subscripts correspond to axes, whilst the superscripts $I\subseteq \{1,\ldots,n\}$ are in correspondence with vertices of the $n$-cube.
Thus the variable $x_i^I$ is associated with an edge of the $n$-cube, aligned in the same direction as axis $i$, and connected to vertex $I$.
The variable $y_i$ is associated with characteristic $i$ of the $n$-cube, which is the $(n-1)$-dimensional section that bisects all edges oriented in direction $i$.
The constraint on variables (\ref{cc}) is associated with the quad formed by the vertices $I$, $I\cup\{i\}$, $I\cup\{j\}$ and $I\cup\{i,j\}$.

We remark that in Proposition \ref{YBP}, the case $n=3$ is important because it is equivalent to the Yang-Baxter property, and subsequent cases $n>3$ are a consequence of the quad symmetry of the defining system and this property.
Case $n=2$ is a statement about the rationality of the defining system of equations.

The next step we take in this paper is the introduction of a second class of integrable dynamics, with a lattice geometry that is different from the $n$-cube, which emerges from elementary considerations of the same family.

\section{Edge-oriented tetrahedral consistency}\label{OTC}
We establish a closure/associativity property in the family of idempotent biquadratic correspondences that share a common orbit.
This can be seen as a consistency feature of the rational three-orbit constraint (\ref{ratmap}).
It allows to interpret it as a discrete integrable dynamical system whose domain is the edge-oriented $n$-simplex.

The relevant closure property has the combinatorics of the tetrahedron with oriented edges, thus in the following proposition the 
indices are associated with vertices and variables with edges, see Figures \ref{tetra}(a) and \ref{tetra}(b).
In this setting, traversing a triangular face of the tetrahedron in one direction yields a triplet of variables, whilst traversing it in the other direction yeilds a second triplet of variables, and this pair of triplets define the idempotent biquadratic polynomial (Lemma \ref{exists}) associated with the face.
\begin{prop}\label{octo}
Denote twelve variables taking values in $\mathbb{C}\cup\{\infty\}$ by
\begin{equation}
x_{ij}, \quad i,j\in\{1,2,3,4\},\ |\{i,j\}|=2,\label{otv}
\end{equation}
and take
\begin{equation}
x_{ij},x_{jk},x_{ki}, \quad x_{ik},x_{kj},x_{ji}, \quad i,j,k\in\{1,2,3,4\}, \ |\{i,j,k\}|=3,\label{ottp}
\end{equation}
to be pairs of orbits that define (cf. Lemma \ref{exists}) four idempotent biquadratic polynomials. 
If there exists an orbit common amongst any three of these, then it is common amongst all four.
\end{prop}
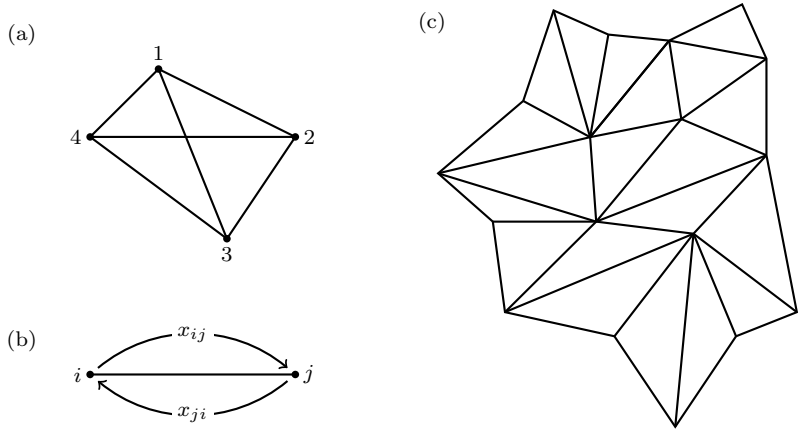
\begin{figure}[t]
\begin{center}
\begin{tikzpicture}[thick,scale=0.9]
  \draw (0,7) node{(a)};
  \coordinate (1) at (2,6.5) {}; 
  \coordinate (2) at (4,5.5) {}; 
  \coordinate (3) at (3,4) {}; 
  \coordinate (4) at (1,5.5) {}; 
  \draw (1) node[above]{1} -- (2) node[right]{2} -- (3) node[below]{3} -- (4) node[left]{4} -- cycle;
  \draw (1) node[circle,fill=black,inner sep=1]{} -- (3) node[circle,fill=black,inner sep=1]{};
  \draw (2) node[circle,fill=black,inner sep=1]{} -- (4) node[circle,fill=black,inner sep=1]{};
  \draw (0,2.5) node{(b)};
  \coordinate (i) at (1,2);
  \coordinate (j) at (4,2);
  \node[circle,fill=black,inner sep=1] (I) at (i) {};
  \node[circle,fill=black,inner sep=1] (J) at (j) {};
  \draw (i) node[left]{$i$} -- (j) node[right]{$j$};
  \draw (I) edge[->, bend left=40,shorten <= 2pt, shorten >= 2pt] node[fill=white,inner sep=2]{$x_{ij}$} (J);
  \draw (I) edge[<-, bend right=40,shorten <= 2pt, shorten >= 2pt] node[fill=white,inner sep=2]{$x_{ji}$} (J);
\end{tikzpicture}
\hspace{30pt}
\begin{tikzpicture}[thick,scale=0.8]
  \tikzstyle{every node}=[inner sep=0pt]
  \coordinate (1) at (2.5,7.2); 
  \coordinate (2) at (3.4,6.8);
  \coordinate (3) at (4.4,6.7); 
  \coordinate (4) at (5.6,7.3);
  \coordinate (5) at (2.0,5.7);
  \coordinate (6) at (3.1,5.1);
  \coordinate (7) at (4.6,5.4);
  \coordinate (8) at (6.0,6.4);
  \coordinate (9) at (0.6,4.5);
  \coordinate (10) at (1.5,3.7);
  \coordinate (11) at (3.2,3.7);
  \coordinate (12) at (4.8,3.5);
  \coordinate (13) at (6.0,4.8);
  \coordinate (14) at (1.7,2.2);
  \coordinate (15) at (3.5,1.8);
  \coordinate (16) at (4.5,0.3);
  \coordinate (17) at (5.5,1.8);
  \coordinate (18) at (6.5,2.2);
  \draw (1) -- (2) -- (3) -- (4) -- (8) -- (13) -- (18) -- (17) -- (16) -- (15) -- (14) -- (10) -- (9) -- (5) -- (1) -- cycle;
  \draw (1) -- (6) -- (11) -- (14);
  \draw (2) -- (6) -- (9);
  \draw (5) -- (6) -- (3);
  \draw (3) -- (7) -- (8);
  \draw (13) -- (7) -- (11);
  \draw (12) -- (18);
  \draw (3) -- (6);
  \draw (6) -- (7);
  \draw (3) -- (8);
  \draw (14) -- (12) -- (13);
  \draw (13) -- (11);
  \draw (12) -- (11) -- (10);
  \draw (9) -- (11);
  \draw (12) -- (17);
  \draw (12) -- (16);
  \draw (12) -- (15);
  \draw (0.5,7) node{(c)};
\end{tikzpicture}
\end{center}
\caption{(a) Tetrahedron, (b) one edge and two variables, one variable is assigned to each orientation of the edge, (c) a planar triangle tessellation.}
\label{tetra}
\end{figure}

More algebraic intuition is offered by interpreting this property as associativity.
Consider the binary operation
\begin{equation}
(x_0,y_0)\cdot(x_1,y_1) := (y_2,x_2)\label{binop}
\end{equation}
defined by system (\ref{ratmap}) when $z_0$, $z_1$ and $z_2$ are considered fixed.
The commutativity of the defined operation is clear from the permutation symmetry of the defining system, the associativity can verified directly by calculation.
This is equivalent to verifying Proposition \ref{octo}, which can be seen as follows.

The proposition states that if the three-orbit constraint (\ref{ratmap}) is satisfied for three of the idempotents, then it is satisfied for the fourth.
Taking the common third orbit to be $z_0$, $z_1$, $z_2$, three of the constraints written in terms of the above binary operation are as follows
\begin{equation}
\begin{split}
(x_{14},x_{41}) = (x_{13},x_{31})\cdot(x_{34},x_{43}),\\
(x_{24},x_{42}) = (x_{21},x_{12})\cdot(x_{14},x_{41}),\\
(x_{23},x_{32}) = (x_{21},x_{12})\cdot(x_{13},x_{31}).
\end{split}\label{first_three}
\end{equation}
Assuming these hold one sees by substitution that the fourth constraint,
\begin{equation}
(x_{24},x_{42}) = (x_{23},x_{32})\cdot(x_{34},x_{43}),\label{fourth}
\end{equation}
is equivalent to the associativity condition
\begin{equation}
(x_{21},x_{12})\cdot[(x_{13},x_{31})\cdot(x_{34},x_{43})] =
[(x_{21},x_{12})\cdot(x_{13},x_{31})]\cdot(x_{34},x_{43}).
\end{equation}

That the equations represented by (\ref{fourth}) are a consequence of those represented by (\ref{first_three}) means this is a consistency property similar in spirit to the known integrability feature of the quadrirational models, namely the consistency on a cube, or Yang-Baxter property.
Here equations on three faces of the tetrahedron (\ref{first_three}) imply the equation on the fourth face (\ref{fourth}).

Just as consistency on the cube leads naturally to the $n$-cube lattice geometry (Proposition \ref{YBP}), the edge-oriented tetrahedral consistency leads to lattice geometry of the $n$-simplex.
\begin{prop}\label{cons}
Denote $n(n+1)$ variables by
\begin{equation}
x_{ij}, \quad i,j\in\{1,\ldots,n+1\},\ |\{i,j\}|=2,
\end{equation}
choose any $z_0,z_1,z_2\in\mathbb{C}\cup\{\infty\}$, and impose the idempotent-biquadratic three-orbit constraint (cf. Propositions \ref{constraint2} and \ref{conform}) on all sets of variables
\begin{equation}
x_{ij},x_{jk},x_{ki}, \quad x_{ik},x_{kj},x_{ji}, \quad z_0,z_1,z_2,\label{tobs}
\end{equation}
for
\begin{equation}
i,j,k\in\{1,\ldots,n+1\}, \ |\{i,j,k\}|=3.
\end{equation}
Then variables remaining from the whole set are rational functions of the unconstrained subset 
\begin{equation}x_{i1},x_{1i}, \quad i\in\{2,\ldots,n+1\}.\label{simplexivp}\end{equation}
\end{prop}
Here indices correspond to vertices of the $n$-simplex, variables are associated with edges, and the constraints are on faces.
The given unconstrained subset of variables are associated with edges connected to vertex $1$.
Tessellation by triangles of any two-dimensional surface (see Figure \ref{tetra}(c)) can be considered as a sub-case of this domain, which is clear because the $n$-simplex can be recovered by adding edges until all pairs of vertices are joined.

We remark that case $n=2$ of Proposition \ref{cons} is a statement about rationality of the defining system, case $n=3$ is equivalent to the consistency property of Proposition \ref{octo}, and $n>3$ is a consequence of this consistency combined with the triangle symmetry of the defining system.

The system described in Proposition \ref{cons} can be directly integrated, the general solution in terms of the initial data is given in terms of the previously defined binary operation (\ref{binop}) by $(x_{ij},x_{ji}) = (x_{i1},x_{1i})\cdot(x_{1j},x_{j1})$, $i,j\in\{2,\ldots,n+1\}$.
\begin{table}[t]
\begin{center}
\begin{tabular}{lll}
\vspace{2pt}
$z_0,z_1,z_2$ & $\varphi(a,b)$& $\varphi(a,b)\cdot\varphi(c,d)$\\
\hline
$0,1,\infty$ & $(a(1-b)/(a-b),(1-b)/(a-b))$ & $\varphi(ac,bd)$\\
$0,\infty,\infty$ & $(a/(b-1),ab/(b-1))$ & $\varphi(a+c,bd)$\\
$\infty,\infty,\infty$ & $(a/b+b,a/b-b)$ & $\varphi(a+c,b+d)$
\end{tabular}
\end{center}
\caption{Decomposing the binary operation into addition and multiplication.}
\label{list}
\end{table}
The binary operation (\ref{binop}) defined by (\ref{ratmap}) is of itself very simple, it can be decomposed into addition and multiplication.
Substitutions $\varphi$ achieving this are listed in Table \ref{list} for the three canonical choices of the fixed third orbit that were used previously in Proposition \ref{conform}.

This example of a discrete system with lattice geometry of the $n$-simplex is thus solved completely.
More generally, the relationship between the edge-oriented tetrahedral consistency and associativity of a binary operation means that a sufficient ingredient for construction of integrable systems of the same kind, is nothing more than an abelian group $A$ in which variables $x_{ij}$ take values.
The system of equations (on faces of the $n$-simplex) $x_{ij}=x_{ik}x_{kj}$, $i,j,k\in\{1,\ldots,n+1\}$, is solved by $x_{ij}=a_ia_j^{-1}$ for some freely chosen $a_1,\ldots,a_{n+1}\in A$.
In isolation, this example of such system is therefore quite trivial, however it plays a key role in the narrative here.
In particular, the defining equations are the same as those that define the system on the $n$-cube described in Proposition \ref{YBP}, and it is therefore natural to ask how the $n$-cube and $n$-simplex lattice geometries are related.
Answering this question is in fact the main technical focus of the remainder of this paper.
We begin in the following section with a local consideration of this problem, using combinatorial arguments to unify the underlying cubic consistency and edge-oriented tetrahedral consistency properties.

\section{5-simplex consistency}\label{5S}
We continue to study the family of idempotent biquadratic polynomials that share a common fixed orbit.
The quadrirationality of the three-orbit constraint established in Section \ref{IB2} led to equivalence with the Yang-Baxter maps, but Section \ref{OTC} established a second consistency feature of the constraint which is in addition to the usual Yang-Baxter property.
The relationship between these two properties is explained here by a stronger consistency feature that unifies them.

\subsection{Completeness based on symmetry}
The defining system of equations, namely the idempotent biquadratic three-orbit constraint in which the third orbit is constant (cf. Propositions \ref{constraint2} and \ref{conform}), is of a general class that we will call a {\it rational triplet-pair system}.
\begin{definition} \label{tps}
By a rational triplet-pair system we mean a system of equations relating a pair of triplets of variables \[x_0,x_1,x_2,\quad y_0,y_1,y_2,\] that rationally determines one variable from each triplet from the remaining four variables, and which is invariant under the symmetry group of the triplet-pair, i.e., permutations of the variables that send the set $\{\{x_0,x_1,x_2\},\{y_0,y_1,y_2\}\}$ to itself.
\end{definition}
It is useful to broaden the notion of domain for this class of system to a more combinatorial setting.
In general the domain is a triplet-pair arrangement.
\begin{definition}\label{tp}
(i) A triplet-pair is a disjoint pair of sets, each set containing three variables.
(ii) A triplet-pair arrangement is a set of variables arranged into triplet-pairs.
\end{definition}

Propositions \ref{YBP} and \ref{cons} involve a rational triplet-pair system imposed on particular triplet-pair arrangements.
\begin{definition} \label{CTP}
The triplet-pair arrangement of the $n$-cube is a set of $n+n2^{(n-1)}$ variables
\begin{equation}\label{ecvars}
y_i, \quad x_i^I, \qquad i \in \{1,\ldots,n\}, \ I\subseteq \{1,\ldots,n\} \setminus \{i\},
\end{equation}
arranged into $n(n-1)2^{(n-3)}$ triplet-pairs
\begin{equation}
y_i,x_j^I,x_j^{I\cup\{i\}}, \quad y_j,x_i^I,x_i^{I\cup\{j\}}, \quad i,j\in\{1,\ldots,n\}, \ i\neq j, \ I\subseteq \{1,\ldots,n\}\setminus\{i,j\}.
\end{equation}
\end{definition}
\begin{definition} \label{EOTTP}
The triplet-pair arrangement of the edge-oriented $n$-simplex is a set of $n(n+1)$ variables
\begin{equation}\label{evars}
x_{ij}, \quad i,j\in\{1,\ldots,n+1\},\ |\{i,j\}|=2,
\end{equation}
arranged into $(n-1)n(n+1)/6$ triplet-pairs
\begin{equation}
x_{ij},x_{jk},x_{ki}, \quad x_{ji},x_{kj},x_{ik}, \quad i,j,k\in\{1,\ldots,n+1\},\ |\{i,j,k\}|=3.
\end{equation}
\end{definition}

The most significant feature of such arrangements is the associated symmetry group.
\begin{definition}\label{symdef}
A symmetry of a triplet-pair arrangement is a permutation of the variables which also permutes the triplet-pairs.
\end{definition}
The triplet-pair arrangements in Definitions \ref{CTP} and \ref{EOTTP} have the same group of symmetries as their respective associated geometric objects, namely the $n$-cube and the $n$-simplex with oriented edges.
However, a single triplet-pair has greater symmetry than the geometric object with which it is associated, in the first case that is a quad, and in the second case a triangle with oriented edges.
In the present setting we adopt the principle that this is a deficiency.
To overcome it we seek to extend the above arrangements, such extension will be termed symmetry-complete.
\begin{definition}\label{sc}
A triplet-pair arrangement is symmetry-complete when every symmetry of every triplet-pair is the restriction of a symmetry of the whole arrangement.
\end{definition}
It is natural to require any such extension be minimal, that is to ask for the {\it smallest} symmetry-complete extension that contains the original.
The question of whether a finite completion of this kind exists is purely combinatorial.

In this section we focus on the triplet-pair arrangements in Definitions \ref{CTP} and \ref{EOTTP} in the important case $n=3$, that is the cube and edge-oriented tetrahedron arrangements.
In these cases a finite symmetry-complete extension exists, and is in fact the same in both cases.
We begin by introducing the symmetrised arrangement, it has the combinatorics of the 5-simplex explained by Figure \ref{5simplex}.
Thus in the following definition we use the natural coordinatisation in which indices are associated with vertices of the 5-simplex, the variables are associated with edges and the triplet-pairs with pairs of vertex-disjoint (or polar) triangles.
\begin{figure}[t]
\begin{center}
\begin{tikzpicture}[thick,scale=1.2]
  \draw (0:1.5) -- (60:1.5) -- (120:1.5) -- (180:1.5) -- (240:1.5) -- (300:1.5) -- (0:1.5) -- cycle;
  \draw (0:1.5) -- (120:1.5);
  \draw (0:1.5) -- (180:1.5);
  \draw (0:1.5) -- (240:1.5);
  \draw (60:1.5) -- (180:1.5);
  \draw (60:1.5) -- (240:1.5);
  \draw (60:1.5) -- (300:1.5);
  \draw (120:1.5) -- (240:1.5);
  \draw (120:1.5) -- (300:1.5);
  \draw (180:1.5) -- (300:1.5);
  \tikzstyle{every node}=[circle,fill=black,inner sep=0pt,minimum size=3pt]
  \node at (0:1.5) {};
  \node at (60:1.5) {};
  \node at (120:1.5) {};
  \node at (180:1.5) {};
  \node at (240:1.5) {};
  \node at (300:1.5) {};
\end{tikzpicture}
\hspace{20pt}
\begin{tikzpicture}[thick,scale=1.2]
  \draw[gray!20!white] (0:1.5) -- (60:1.5) -- (120:1.5) -- (180:1.5) -- (240:1.5) -- (300:1.5) -- (0:1.5) -- cycle;
  \draw[gray!20!white] (0:1.5) -- (120:1.5);
  \draw[gray!20!white] (0:1.5) -- (180:1.5);
  \draw[gray!20!white] (0:1.5) -- (240:1.5);
  \draw[gray!20!white] (60:1.5) -- (180:1.5);
  \draw[gray!20!white] (60:1.5) -- (240:1.5);
  \draw[gray!20!white] (60:1.5) -- (300:1.5);
  \draw[gray!20!white] (120:1.5) -- (240:1.5);
  \draw[gray!20!white] (120:1.5) -- (300:1.5);
  \draw[gray!20!white] (180:1.5) -- (300:1.5);
  \draw (0:1.5) -- (240:1.5) -- (180:1.5) -- cycle;
  \draw (60:1.5) -- (120:1.5) -- (300:1.5) -- cycle;
  \tikzstyle{every node}=[circle,fill=black,inner sep=0pt,minimum size=3pt]
  \node at (0:1.5) {};
  \node at (60:1.5) {};
  \node at (120:1.5) {};
  \node at (180:1.5) {};
  \node at (240:1.5) {};
  \node at (300:1.5) {};
\end{tikzpicture}
\end{center}
\caption{The first diagram is the 5-simplex, fifteen variables are assigned to the fifteen edges, and ten triplet-pairs are assigned to the ten vertex-disjoint pairs of triangles, or {\it polar} triangles, a typical example of which is illustrated in the second diagram.}
\label{5simplex}
\end{figure}
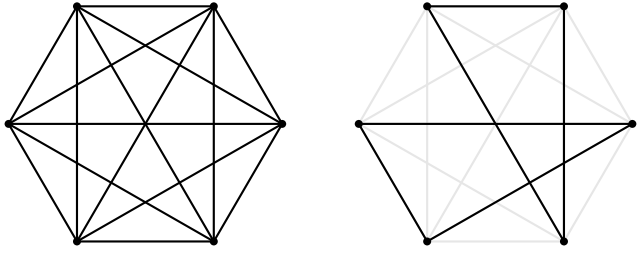
\begin{definition}\label{5SPC}
The triplet-pair arrangement of the 5-simplex is a set of fifteen variables
\begin{equation}
w_{ij} = w_{ji}, \quad i,j\in\{1,2,3,4,5,6\}, \ |\{i,j\}|=2,\label{simpvars}
\end{equation}
arranged into ten triplet-pairs
\begin{equation}
w_{ij},w_{jk},w_{ki}, \quad w_{lm},w_{mn},w_{nl}, \quad \{i,j,k,l,m,n\} = \{1,2,3,4,5,6\}.\label{triangles}
\end{equation}
\end{definition}
The combinatorial result can then be stated as follows.
\begin{prop}\label{5s}
The triplet-pair arrangement of the 5-simplex (Definition \ref{5SPC}) is the smallest symmetry-complete extension of the triplet-pair arrangement of both the cube (case $n=3$ of Definition \ref{CTP}) and the edge-oriented tetrahedron (case $n=3$ of Definition \ref{EOTTP}).
\end{prop}
\begin{proof}
That the triplet-pair arrangement of the 5-simplex is symmetry-complete in the sense of Definition \ref{sc} can be seen from Figure \ref{5simplex}, the symmetry group of the 5-simplex is $S_6$ which acts naturally on its vertices.
By inspection this group acts transitively on the triangle-pairs, whilst the subgroup that stabilizes any particular triangle-pair is the complete symmetry group of that pair.

Because all arrangements are finite, it is possible to directly verify that the 5-simplex arrangement contains the cube and edge-oriented tetrahedron arrangements.
We proceed by giving explicit associations between variables of the arrangements.
\begin{equation}
\begin{array}{lll}
y_1 & y_2 & y_3 \\[0.1em]
x_1^{\{\}} & x_2^{\{\}} &x_3^{\{\}} \\[0.1em]
x_1^{\{2\}} & x_2^{\{3\}} &x_3^{\{1\}} \\[0.1em]
x_1^{\{3\}} & x_2^{\{1\}} &x_3^{\{2\}} \\[0.1em]
x_1^{\{2,3\}} & x_2^{\{1,3\}} &x_3^{\{1,2\}}
\end{array}
\quad
\rightarrow
\qquad
\begin{array}{lll}
w_{12} \ & w_{34} \ & w_{56} \\[0.3em]
w_{45} \ & w_{16} \ & w_{23} \\[0.3em]
w_{35} \ & w_{15} \ & w_{13} \\[0.3em]
w_{46} \ & w_{26} \ & w_{24} \\[0.3em]
w_{36} \ & w_{25} \ & w_{14}
\end{array}
\label{cfs}
\end{equation}
Here the variables on the left are from the cube arrangement, and they are put in correspondence with the variables of the 5-simplex arrangement on the right.
It is straightforward to verify that each triplet-pair of the cube arrangement (case $n=3$ of Definition \ref{CTP}) corresponds to a triplet-pair in the 5-simplex arrangement (Definition \ref{5SPC}).

Similar associations are possible for the edge-oriented tetrahedral arrangement.
\begin{equation}
\begin{array}{lll}
x_{21} \ & x_{41} \ & x_{31} \\[0.3em]
x_{12} \ & x_{14} \ & x_{13} \\[0.3em]
x_{24} \ & x_{43} \ & x_{32} \\[0.3em]
x_{23} \ & x_{42} \ & x_{34} \\[0.3em]
- & - & - 
\end{array}
\quad
\rightarrow
\qquad
\begin{array}{lll}
w_{12} \ & w_{34} \ & w_{56} \\[0.3em]
w_{45} \ & w_{16} \ & w_{23} \\[0.3em]
w_{35} \ & w_{15} \ & w_{13} \\[0.3em]
w_{46} \ & w_{26} \ & w_{24} \\[0.3em]
w_{36} \ & w_{25} \ & w_{14}
\end{array}
\label{otfs}
\end{equation}
Once again, via this embedding, it is straightforward to verify that each triplet-pair of the edge-oriented tetrahedron arrangement (case $n=3$ of Definition \ref{EOTTP}) corresponds to a triplet-pair of the 5-simplex arrangement (Definition \ref{5SPC}).

To complete the proof it remains to establish that there are no smaller symmetry-complete triplet-pair arrangements that contain either the arrangement of the cube or of the edge-oriented tetrahedron.
This can be done by considering the intersections between triplet-pairs.
In both the cube and edge-oriented tetrahedron arrangements, all intersections are of the same nature, triplet-pairs intersect on two variables - one from each triplet.
For a hypothetical symmetry-complete extension it is necessary that for any participating triplet-pair and any two variables chosen one from each of its triplets, that there exist another triplet-pair intersecting the first on those variables.
For the 5-simplex arrangement, a triplet-pair intersecting in this way exists and is in fact {\it unique}, and furthermore {\it all} intersections are of this kind.
The extended arrangement therefore cannot be any smaller, for if it were it would omit some of the necessary intersections.
\end{proof}

\subsection{The stronger consistency property}
Proposition \ref{5s} established the triplet-pair arrangement of the 5-simplex as a natural combinatorial extension of the cube and edge-oriented tetrahedron arrangements.
We now introduce the corresponding extended consistency property as follows.
\begin{definition} \label{fsc}
If a rational triplet-pair system (Definition \ref{tps}) imposed on triplet-pairs of the 5-simplex arrangement (Definition \ref{5SPC}) leaves the variables 
\begin{equation}
w_{12}, w_{34}, w_{56}, w_{45}, w_{16}, w_{23},\label{us}
\end{equation}
unconstrained, then it is called 5-simplex consistent.
\end{definition}
It is straightforward to see that this definition gives the desired extension, that is:
\begin{corollary}[following from Proposition \ref{5s}]\label{xcp}
If a rational triplet-pair system is 5-simplex consistent (Definition \ref{fsc}), then it is consistent when imposed on the triplet-pairs of the cube arrangement (case $n=3$ of Definition \ref{CTP}) in the sense that it leaves the variables
\begin{equation}
x_{1}^{\{\}},x_{2}^{\{\}},x_{3}^{\{\}},y_1,y_2,y_3,\label{cubeid}
\end{equation}
unconstrained.
And it is also consistent on the triplet-pair arrangement of the edge-oriented tetrahedron (case $n=3$ of Definition \ref{EOTTP}), in the sense that it leaves the variables
\begin{equation}
x_{12},x_{13},x_{14},x_{21},x_{31},x_{41},\label{tetrahedronid}
\end{equation}
unconstrained.
(These are exactly the unconstrained variables in case $n=3$ of Propositions \ref{YBP} and \ref{cons}.)
\end{corollary}
\begin{proof}
This can be seen by inspection of the previously established correspondences (\ref{cfs}) and (\ref{otfs}) that embed the cube and edge-oriented tetrahedron triplet-pair arrangements in the 5-simplex one.
It is sufficient to observe that the six variables in the unconstrained subset (\ref{us}) correspond to the six variables (\ref{cubeid}) and (\ref{tetrahedronid}) respectively.
\end{proof}
The examples at hand are indeed 5-simplex consistent:
\begin{prop}\label{pdds}
With constant third orbit $z_0,z_1,z_2\in\mathbb{C}\cup\{\infty\}$, the idempotent-biquadratic three-orbit constraint (cf. Propositions \ref{constraint}, \ref{constraint2} and \ref{conform}) defines a rational triplet-pair system (Definition \ref{tps}) which is 5-simplex consistent (Definition \ref{fsc}).
\end{prop}
\begin{proof}
This can be verified by calculation.
Clearly from Corollary \ref{xcp} the calculation extends the ones required to verify the cubic and edge-oriented tetrahedral consistency properties (case $n=3$ of Propositions \ref{YBP} and \ref{cons}).
\end{proof}
Before introducing techniques to unify the $n$-cube and edge-oriented $n$-simplex lattice geometries also for $n>3$, we are now at a point where it is natural to establish contact with the integrable multi-quadratic quad-equations.

\subsection{Connection with the multi-quadratic quad-equations}\label{QUAD}
In Section \ref{F1Q4} it was shown how it is natural to understand the multi-quadratic quad-equation $Q4^*$ as being a consequence of the Yang-Baxter map $F_I$ imposed on the cube (recall Figure \ref{cube}).
The corresponding result generalised to the 5-simplex triplet-pair arrangement, and dealing now with all three possible sizes of the set $\{z_0,z_1,z_2\}$, is (in the notation of Section \ref{F1Q4}) as follows.
\begin{prop}[Connection with canonical forms in \cite{AtkNie}]\label{qq}
Impose on the 5-simplex triplet-pair arrangement (Definition \ref{5SPC}) the rational triplet-pair system defined by the idempotent-biquadratic three-orbit constraint (cf. Propositions \ref{constraint2} and \ref{pdds}).
Then there is a single polynomial equation relating the seven variables
\begin{equation}
w_{12}, w_{45}, w_{36}, w_{35}, w_{46}, w_{34}, w_{56}, \label{varlist}
\end{equation}
which is degree four in $w_{12}$ and degree two in each of the remaining six variables.
Changing notation by writing
\begin{equation}
w_{45}=x, \ w_{36}=\th{x}, \ w_{35}=\wt{x}, \ w_{46}=\wh{x}, \ w_{34}=p, \ w_{56}=q,\label{quadids}
\end{equation}
and making the substitution $(w_{12},z_0,z_1,z_2) = (c,-c,1/c,-1/c)$ for some fixed constant $c$, this polynomial equation coincides with equation $Q4^*$:
\begin{equation}
\begin{split}
&(p-q)[(c^{-2}p-c^2q)(x\wt{x}-\wh{x}\th{x})^2-(c^{-2}q-c^2p)(x\wh{x}-\wt{x}\th{x})^2]\\
&\quad -(p-q)^2[(x+\th{x})^2(1+\wt{x}^2\wh{x}^2)+(\wt{x}+\wh{x})^2(1+x^2\th{x}^2)]\\
&\quad +[(x-\th{x})(\wt{x}-\wh{x})(c^{-1}-cpq)-2(p-q)(1+x\wt{x}\wh{x}\th{x})]\\
&\quad \times [(x-\th{x})(\wt{x}-\wh{x})(c^{-1}pq-c)-2(p-q)(x\th{x}+\wt{x}\wh{x})]=0.
\end{split}
\label{QQ4again}
\end{equation}
On substitution $(w_{12},z_0,z_1,z_2)=(1,-1,\infty,\infty)$, it coincides with $Q3^*_{\delta=1}$:
\begin{equation}
\begin{split}
&(p-q)[p(u\wt{x}-\wh{x}\th{x})^2-q(x\wh{x}-\wt{x}\th{x})^2]-(p-q)^2[(x+\th{x})^2+(\wt{x}+\wh{x})^2]\\
&+[(x-\th{x})(\wt{x}-\wh{x})-p+q][(x-\th{x})(\wt{x}-\wh{x})(pq-1)-2(p-q)(x\th{x}+\wt{x}\wh{x})]=0.
\end{split}
\label{QQ3}
\end{equation}
And on substitution $(w_{12},z_0,z_1,z_2)=(0,\infty,\infty,\infty)$, it coincides with $Q2^*$:
\begin{equation}
\begin{split}
&(p-q)[p(x\wt{x}-\wh{x}\th{x})(x+\wt{x}-\wh{x}-\th{x})-q(x\wh{x}-\wt{x}\th{x})(x+\wh{x}-\wt{x}-\th{x})] \\
&+(x-\th{x})(\wt{x}-\wh{x})[p(x-\wh{x})(\wt{x}-\th{x})-q(x-\wt{x})(\wh{x}-\th{x})-pq(p-q)]=0.
\end{split}
\label{QQ2}
\end{equation}
\end{prop}
\begin{proof}
The described polynomial relation can be obtained by straightforward elimination of all variables except those listed in (\ref{varlist}) from the system of equations in question.
It is a lengthy expression, but can be written in the following form when $|\{z_0,z_1,z_2\}|=3$,
\begin{multline}
\psi(z_0,z_1,z_2) + \psi(z_1,z_2,z_0)+\psi(z_2,z_0,z_1)+\\\psi(z_2,z_1,z_0)+\psi(z_0,z_2,z_1)+\psi(z_1,z_0,z_2)=0,\label{phiphi}
\end{multline}
where 
\begin{equation}
\begin{split}
\psi(z_0,z_1,z_2) & = (w_{12}-z_2)(z_0-z_1)\CR(w_{34},z_1,w_{56},z_2)\CR(w_{34},z_2,w_{56},z_0)\\
&\quad \times\big(\!\CR(w_{45},w_{12},w_{36},z_2)-\CR(w_{46},z_1,w_{35},z_0)\big)\\
&\quad \times\big(\!\CR(w_{45},z_1,w_{36},z_0)-\CR(w_{46},w_{12},w_{35},z_2)\big),
\end{split}
\label{phidef}
\end{equation}
and $\CR$ denotes the cross-ratio
\begin{equation}
\CR(a,b,c,d):=\frac{(a-b)(c-d)}{(a-c)(b-d)}.
\end{equation}
Starting from expression (\ref{phiphi}), the substitutions indicated in the proposition should be taken as limits in the second two cases. 
Specifically, the leading order term that remains in (\ref{phiphi}) upon substitution of $w_{12}=0$ and $(z_0,z_1,z_2)\rightarrow(z_0/\epsilon,z_1/\epsilon,z_2/\epsilon)$ and taking the limit $\epsilon\rightarrow 0$ is equivalent to (\ref{QQ2}) modulo identifications (\ref{quadids}).
The substitution $w_{12}=1$, $(z_0,z_1,z_2)\rightarrow(-1,z_1/\epsilon,z_2/\epsilon)$ leads similarly to (\ref{QQ3}).
And the substitution $w_{12}=c$, $(z_0,z_1,z_2)=(-c,1/c,-1/c)$ leads to (\ref{QQ4again}).
\end{proof}

The subset of variables (\ref{varlist}) corresponds to a singled-out edge of the 5-simplex connecting vertices $1$ and $2$, and the edges of its polar tetrahedron formed by the remaining vertices $3,4,5,6$, see Figure \ref{tetrahedron}.
The tetrahedral permutation symmetry of the polynomial equations (\ref{QQ4again}), (\ref{QQ3}) and (\ref{QQ2}) is a new observation which can be verified by calculation, like in the case of systems in Proposition \ref{YBid}, this additional symmetry is not obvious from the given expressions.
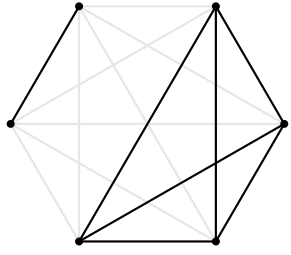
\begin{figure}[t]
\begin{center}
\begin{tikzpicture}[thick,scale=1.2]
  \draw[gray!20!white] (0:1.5) -- (60:1.5) -- (120:1.5) -- (180:1.5) -- (240:1.5) -- (300:1.5) -- (0:1.5) -- cycle;
  \draw[gray!20!white] (0:1.5) -- (120:1.5);
  \draw[gray!20!white] (0:1.5) -- (180:1.5);
  \draw[gray!20!white] (0:1.5) -- (240:1.5);
  \draw[gray!20!white] (60:1.5) -- (180:1.5);
  \draw[gray!20!white] (60:1.5) -- (240:1.5);
  \draw[gray!20!white] (60:1.5) -- (300:1.5);
  \draw[gray!20!white] (120:1.5) -- (240:1.5);
  \draw[gray!20!white] (120:1.5) -- (300:1.5);
  \draw[gray!20!white] (180:1.5) -- (300:1.5);
  \draw (180:1.5) -- (120:1.5);
  \draw (300:1.5) -- (60:1.5);
  \draw (0:1.5) -- (240:1.5);
  \draw (300:1.5) -- (0:1.5) -- (60:1.5) -- (240:1.5) -- cycle;
  \tikzstyle{every node}=[circle,fill=black,inner sep=0pt,minimum size=3pt]
  \node at (0:1.5) {};
  \node at (60:1.5) {};
  \node at (120:1.5) {};
  \node at (180:1.5) {};
  \node at (240:1.5) {};
  \node at (300:1.5) {};
\end{tikzpicture}
\end{center}
\caption{An edge of the 5-simplex and its polar tetrahedron.}
\label{tetrahedron}
\end{figure}
It is clear from Figure \ref{tetrahedron} that there are in fact 15 such seven-variable polynomial equations emerging as a consequence of the system described in Proposition \ref{qq}, corresponding to the 15 distinct edge-tetrahedron polar-pairs of the 5-simplex.
This system of 15 equations involves a total of 15 variables, and leaves at least 6 variables unconstrained according to Proposition \ref{pdds}, which is therefore a non-trivial consistency property for the seven-variable polynomial equation.
One can notice that, with respect to this consistency property, the specialisations leading to canonical forms (\ref{QQ4again}), (\ref{QQ3}) and (\ref{QQ2}) are unnatural in the sense that they break the symmetry of the 5-simplex by singling out variable $w_{12}$.
This is despite the fact that in the original context of the quad-equations, this specialisation is without loss of generality.
From the point of view of the theory developed here, these canonical forms are therefore parameter-deficient, and it is preferable to consider instead the homogeneous form (\ref{phiphi})-(\ref{phidef}), possibly with constant values chosen for $z_0,z_1,z_2$, but considering $w_{12}$ as being on an equal footing with the other variables.

We remark that Proposition \ref{qq} says there is exactly one polynomial relation between the seven variables (\ref{varlist}).
The immediate corollary is that any six of these variables are unconstrained.
Clearly then, the specified system of equations defines a rational transformation from the six variables (\ref{us}) to any six of the variables from (\ref{varlist}), and this is non-degenerate, meaning the inverse is only finitely multivalued.
Due to this non-degeneracy, all solutions of the seven-variable equation can be obtained rationally by expressing the participating variables (\ref{varlist}) in terms of the alternative set (\ref{us}).
This is a {\it rational reformulation}, existence of which is a non-trivial property of the polynomial equation, which is connected with the discriminant-factorisation at the basis of the original construction in \cite{AtkNie}.

The most salient aspect of Proposition \ref{qq} is the confirmation that the featured higher degree equations $Q4^*$, $Q3^*_{\delta=1}$ and $Q2^*$, are most naturally considered as derived equations from the simpler underlying rational system described in Proposition \ref{pdds}.
The subsequent considerations of this paper will therefore concern only the rational system.
The main question remaining about the multi-quadratic quad-equations is the nature of the relationship between the above described 5-simplex consistency and the usual cubic-consistency property.
The answer to this question is again combinatorial, it is connected with the symmetry-complete extension of arrangements in Definitions \ref{CTP} and \ref{EOTTP}.

\section{The vertex-map group}\label{FH}
In the previous section we introduced the 5-simplex consistency property, unifying the cubic and edge-oriented tetrahedral consistency.
To address the question of the relation between the $n$-cube and edge-oriented $n$-simplex lattice geometries also for $n>3$, we exploit a natural relation between lattice systems and birational groups, which allows application of group analysis to the problem.
The approach to studying lattice systems through birational groups is familiar from the setting of Yang-Baxter maps, as well as in the theory of the discrete Painlev\'e equations.

\subsection{The general notion of an induced birational group}\label{BRF}
The general setting for the induced birational group is a set of variables $V$ on which there are imposed equations such that, (i) the equations are invariant under permutations of the variables forming a group $G$ and, (ii) the equations determine rationally all variables remaining in $V$ from some unconstrained subset $\{x_1,x_2,\ldots\}\subset V$ (which may also be called the initial data set).

This setting allows to place the unconstrained subset of variables in an array $X=[x_1,x_2,\ldots]$, and define a function on $g\in G$ as $X(g)=[g(x_1),g(x_2),\ldots]$.
It is then clear that, by assumption (ii), the variables in array $X(g)$ may be expressed as rational functions of those in $X=X(\id)$.
This is what we call the induced rational action of $g$, and we write
\begin{equation}\label{Xind}
X(g) = g(X(\id)), \quad g\in G.
\end{equation}
We use the symbol $g$ also for the induced rational action (on the right-hand-side in (\ref{Xind})) because it will always be distinguishable by context.
Given the induced rational action of some fixed $g\in G$, it is clear that $X(gh)=g(X(h))$ for all $h\in G$.
This is because of (i), i.e., the equations defining the rational action of $g$ are invariant under permutations $h\in G$.
It therefore follows that
\begin{equation}
[gh](X(\id)) = X(gh) = g(X(h)) = g(h(X(\id))), \quad h,g\in G,
\end{equation}
so the composition of rational actions is consistent with the permutation composition.
The induced rational actions therefore constitute a birational representation of $G$.

The birational group is {\it equivalent} to the original system of equations, or {\it non-degenerate}, if application of the group to the initial data array yields all remaining variables in $V$.
That is, when $\{g(x): g\in G, x\in \{x_1,x_2,\ldots\}\} = V$.
All reformulations considered here will be non-degenerate in this sense.

\subsection{Basic participating mappings}\label{fg}
We make here a preliminary construction, introducing a set of primitive mappings which turn out to be very useful to express the emerging birational groups.
They are constructed from a generic rational triplet-pair system (definition \ref{tps}) as follows.
For each $j\in\{1,\ldots,n\}$ impose the rational triplet-pair system on all triplet-pairs
\begin{equation}
y_j,x_i,\wb{x}_i, \quad x_j,y_i,\wb{y}_i, \qquad i \in \{1,\ldots,n\}\setminus\{j\},\label{mc}
\end{equation}
and set $\wb{x}_j=x_j,\wb{y}_j=y_j$.
This means the variables 
\begin{equation}\label{Fvars}
\wb{x}_1,\ldots,\wb{x}_n,\wb{y}_1,\ldots,\wb{y}_n
\end{equation}
are rational functions of the variables 
\begin{equation}\label{Evars}
{x}_1,\ldots,{x}_n,{y}_1,\ldots,{y}_n.
\end{equation}
We then define a rational mapping, $\sigma_j$, acting on a $2\times n$ array of variables as follows
\begin{equation}
\sigma_j:=\flm{x_1,y_1}{x_n,y_n}\mapsto\flm{\wb{x}_1,y_1}{\wb{x}_n,y_n}.\label{vertex}
\end{equation}
Denoting the mapping which transposes the two columns by $\omega$,
\begin{equation}
\omega:=\flm{x_1,y_1}{x_n,y_n}\mapsto\flm{y_1,x_1}{y_n,x_n},\label{omegadef}
\end{equation}
allows, by conjugation, to introduce the natural companion to $\sigma_j$.
We denote it by $\sigma_j^\omega:=\omega\sigma_j\omega^{-1}$:
\begin{equation}
\sigma_j^\omega=\flm{x_1,y_1}{x_n,y_n}\mapsto\flm{x_1,\wb{y}_1}{x_n,\wb{y}_n}.\label{conjvertex}
\end{equation}
Due to the symmetry of the rational triplet-pair system, the mappings $\sigma_j$ and $\sigma_j^\omega$ commute, and their composition yields
\begin{equation}\label{comm}
\sigma_j\sigma_j^\omega=\sigma_j^\omega\sigma_j=\flm{x_1,y_1}{x_n,y_n}\mapsto\flm{\wb{x}_1,\wb{y}_1}{\wb{x}_n,\wb{y}_n}.
\end{equation}

For the rational triplet-pair systems being considered here, that is, the idempotent biquadratic three-orbit constraint in which the third orbit $z_0,z_1,z_2$ is constant, the explicit form of equations determining variables (\ref{Fvars}) as rational functions of variables (\ref{Evars}) are as follows:
\begin{equation}
\begin{split}
&\begin{array}{l}
h(x_j,y_i;z_0,z_1,z_2;x_i,y_j,\wb{x}_i)=0,\\
h(y_j,x_i;z_0,z_1,z_2;y_i,x_j,\wb{y}_i)=0,
\end{array}
\quad i\in\{1,\ldots,n\}\setminus \{j\},\\
&\begin{array}{l}
\wb{x}_j = x_j,\\
\wb{y}_j = y_j,
\end{array}
\end{split}
\label{vvertex}
\end{equation}
where the polynomial $h$ was defined in (\ref{genh}).

\subsection{Vertex-maps of the $n$-cube}\label{CM}
Here we give the result of applying the procedure described in Section \ref{BRF} to a rational triplet-pair system imposed on the triplet-pair arrangement of the $n$-cube, such as the system in Proposition \ref{YBP}.
\begin{prop}\label{BRBC}
A rational triplet-pair system (Definition \ref{tps}) imposed on all triplet-pairs of the $n$-cube arrangement (Definition \ref{CTP})
which leaves variables of the array
\begin{equation}\label{ca}
X:=\flm{x_1^{\{\}},y_1}{x_n^{\{\}},y_n}
\end{equation}
unconstrained, induces a birational representation of the symmetry group of the $n$-cube ($BC_n$) generated by the mappings
\begin{equation}\label{cubegenerators}
\sigma_1, \pi_{1,2},\ldots,\pi_{n-1,n}
\end{equation}
where $\sigma_1$ was the rational mapping defined in (\ref{vertex}), and $\pi_{i,j}$ transposes rows $i$ and $j$ of the array.
\end{prop}
\begin{proof}
The standard generators of the $n$-cube symmetry group consist of one reflection through axis 1, and the $n-1$ reflections that interchange axis $i$ with axis $i+1$, for $i\in\{1,\ldots,n-1\}$, whilst fixing the remaining axes.
The corresponding action on the variables (\ref{ecvars}) clearly leaves the system of equations unchanged.
It is straightforward to verify that the reflection in axis $1$ induces, via initial data array (\ref{ca}), the rational mapping $\sigma_1$ (\ref{vertex}).
The remaining generators only permute the unconstrained subset of variables by interchanging row $i$ and $j$ in (\ref{ca}).
\end{proof}
The mappings (\ref{cubegenerators}) correspond to standard generators of $BC_n$, which has order $2^n n!$.
Note that the variables on characteristics, $y_1,\ldots,y_n$, play a passive role here, the group acts on them only by permutation.

Besides the variables on characteristics, the array $X$ in (\ref{ca}) contains variables on all edges connected to one vertex of the $n$-cube, and for this reason we refer to $X$ as vertex-type initial data and the induced rational mappings as vertex-maps.
We remark that it is often usual to consider the different path-type initial data, which corresponds to the choice of $X$ as
\begin{equation}
X=\left[\begin{array}{ll}
x_1^{\{\}}&, y_1\\
x_2^{\{1\}}&, y_2\\
x_3^{\{1,2\}}&, y_3\\
&\vdots\\
x_n^{\{1,\ldots,n-1\}}&, y_n
\end{array}\right].\label{ccaa}
\end{equation}
From this choice of $X$ one obtains the more usual form of the braid-type mappings \cite{Ves} in which the first generator, corresponding to $\sigma_1$ in (\ref{cubegenerators}) is omitted, the resulting birational group is still equivalent to the underlying system of equations (in the sense described at the end of Section \ref{BRF}).
The choice (\ref{ca}) instead of (\ref{ccaa}) is preferable here.

\subsection{Vertex-maps of the edge-oriented $n$-simplex}\label{SM}
The procedure described in Section \ref{BRF} applied to a rational triplet-pair system imposed on the triplet-pair arrangement of the edge-oriented $n$-simplex, such as the system described in Proposition \ref{cons}, yields the following.
\begin{prop}\label{BRA}
A rational triplet-pair system (Definition \ref{tps}) imposed on all triplet-pairs of the $n$-simplex arrangement (Definition \ref{EOTTP}) which leaves variables of the array
\begin{equation}\label{nsg}
X:= \left[\begin{array}{lcl}
x_{12}&, &x_{21}\\
x_{13}&, &x_{31}\\
&\vdots\\
x_{1(n+1)}\!\!&, &x_{(n+1)1}
\end{array}\right]
\end{equation}
unconstrained, induces a birational representation of the symmetry group of the edge-oriented $n$-simplex ($A_n\times A_1$) generated by the mappings
\begin{equation}\label{simplexgenerators}
\sigma_1\omega\sigma_1, \pi_{1,2},\ldots,\pi_{n-1,n}, \omega
\end{equation}
where $\sigma_1$ and $\omega$ were defined in (\ref{vertex}) and (\ref{omegadef}), whilst $\pi_{i,j}$ is the mapping which transposes rows $i$ and $j$.
\end{prop}
\begin{proof}
The standard generators of the symmetry group of the $n$-simplex act by permuting the $n+1$ vertices (associated with indices), specifically interchanging pairs of vertices $i$ and $i+1$ for each $i\in\{1,\ldots,n\}$.
The induced permutations of edges, and thus permutations of variables (\ref{evars}), clearly leave the system of equations unchanged.
The array $X$ (\ref{nsg}) singles out vertex $1$, because it contains variables on all edges connected to that vertex, so the action of the first generator is distinguished (it is the only one involving vertex $1$).
It is straightforward to verify that the first generator (interchanging vertices $1$ and $2$) induces the rational mapping $\sigma_1\omega\sigma_1$, whereas the remaining generators act by pure permutation on the unconstrained variable set, in particular interchanging rows $i$ and $i+1$ of the array $X$.
The final generator listed in (\ref{simplexgenerators}) is again a pure permutation of the unconstrained variable set, it comes from the overall reversal of edge orientations, which commutes with the rest of the group.
\end{proof}
The first $n$ mappings in (\ref{simplexgenerators}) correspond to standard generators of the $n$-simplex symmetry group, $A_n$, which has order $(n+1)!$.
The last mapping in (\ref{simplexgenerators}), reversing orientation of edges, commutes with the other generators, so the full group is $A_n\times A_1$.
The array (\ref{nsg}) contains variables on all edges attached to vertex $1$, so again it may be called vertex-type initial data.

\subsection{Birational group of the 5-simplex}\label{5SVM}
Here we apply the procedure described in Section \ref{BRF} to the system appearing in the definition of 5-simplex consistency (Definition \ref{fsc}).
\begin{prop}\label{BRA5}
A rational triplet-pair system imposed on all triplet-pairs of the 5-simplex arrangement (Definition \ref{5SPC}) which leaves variables of the array
\begin{equation}\label{gg}
X:= \left[\begin{array}{ll}
w_{12},&w_{54}\\
w_{34},&w_{16}\\
w_{56},&w_{32}
\end{array}\right]
\end{equation}
unconstrained (i.e., is 5-simplex consistent, cf. Definition \ref{fsc}), induces a birational representation of the symmetry group of the 5-simplex ($A_5$) generated by the mappings
\begin{equation}\label{fivegenerators}
\sigma_1^\omega,\sigma_3,\sigma_2^\omega,\sigma_1,\sigma_3^\omega,
\end{equation}
where $\sigma_1,\sigma_2,\sigma_3$ and $\omega$ are the rational mappings defined in (\ref{vertex}) and (\ref{omegadef}) in the particular case $n=3$.
\end{prop}
\begin{proof}
The symmetry group of the 5-simplex acts naturally on the vertices, and vertices correspond to indices of the edge-variables appearing in Definition \ref{5SPC}.
Standard generators of the group interchange pairs of indices: $1\leftrightarrow 2$, $2\leftrightarrow 3$, $3 \leftrightarrow 4$, $4\leftrightarrow 5$, $5\leftrightarrow 6$.
That the corresponding permutations of edge variables in array (\ref{gg}) induce the mappings (\ref{fivegenerators}) is straightforward to verify.
\end{proof}
The rational mappings (\ref{fivegenerators}) correspond to the standard generators of $A_5$.

\subsection{Characterisation of the complete vertex-map group}\label{CVM}
The main freedom we have in obtaining the birational group through the procedure of Section \ref{BRF} is the particular choice of the unconstrained variable set; modulo the symmetry group of the system (i.e. the group $G$ of Section \ref{BRF}) there are usually a few different ones that are possible.
For a system in isolation, such choice is quite irrelevant in the sense that they lead to birationally equivalent groups.
But here we have taken care to highlight the particular choice made, which is the vertex-type unconstrained variable array in the $n$-cube and $n$-simplex birational group construction.
This is a very deliberate choice which has been motivated by the case $n=3$: it exploits the previously established combinatorial unification described in Section \ref{5S}.
In particular, the combinatorial embedding of the cube and edge-oriented tetrahedron triplet-pair arrangements into the 5-simplex one ((\ref{cfs}) and (\ref{otfs}) respectively) establish the correct identification of variables from the unconstrained-variable array.
This manifests in the following fact.
\begin{prop}\label{coincidence}
The birational group of the 5-simplex, i.e., the group generated by mappings (\ref{fivegenerators}), coincides with the birational group obtained by combining generators of the cube and edge-oriented tetrahedron birational groups, respectively (\ref{cubegenerators}) and (\ref{simplexgenerators}), in the case $n=3$.
\end{prop}
\begin{proof}
As mentioned, this is basically due to the combinatorial relationship between the respective consistency properties established in Section \ref{5S}, but it can also be verified directly.
Bearing in mind that $\pi_{i,j}$ acts by interchanging rows $i$ and $j$, it is easy to see from the definitions in Section \ref{fg} that $\pi_{i,j}\sigma_i=\sigma_j\pi_{i,j}$ and $\pi_{i,j}\sigma_i^\omega=\sigma_j^{\omega}\pi_{i,j}$.
It is therefore clear that the group generated by (\ref{fivegenerators}) is a subgroup of the group generated by case $n=3$ of (\ref{cubegenerators}) combined with case $n=3$ of (\ref{simplexgenerators}).
To verify that this is not a proper subgroup, but actually the whole group, it is sufficient to verify that $\omega$, $\pi_{1,2}$, and $\pi_{2,3}$ emerge as some word in the generators (\ref{fivegenerators}).
In turn this is simple to do by making the transition back from the birational group generated by (\ref{fivegenerators}) to the permutation action on vertices of the 5-simplex.
In particular, observe that the action of $\omega$ (\ref{omegadef}) on (\ref{gg}) corresponds to permutation of indices $(1,2,3)\leftrightarrow (4,5,6)$, whilst $\pi_{1,2}$ corresponds to $(1,3,5)\leftrightarrow (4,2,6)$ and $\pi_{2,3}$ corresponds to $(1,3,5)\leftrightarrow (2,6,4)$.
For completeness we express these permutations in the standard generators which allows to write the following.
\begin{eqnarray}
&&\omega=\sigma_2^\omega\sigma_3\sigma_1^\omega\sigma_1\sigma_2^\omega\sigma_3\sigma_3^\omega\sigma_1\sigma_2^\omega,\label{omegarel}\\
&&\pi_{1,2} = (\omega\sigma_1\sigma_2)^3,\label{p12}\\
&&\pi_{2,3} = (\omega\sigma_2\sigma_3)^3.\label{p23}
\end{eqnarray}
Note that the right-hand-side of (\ref{p12}) and (\ref{p23}) are expressed in terms of generators (\ref{fivegenerators}) only modulo the first relation (\ref{omegarel}).
We stress that the equality in (\ref{omegarel}) does not hold if $n>3$.
\end{proof}
Proposition \ref{coincidence} is the basic motivating fact for how to proceed, it shows that the birational group of the 5-simplex can be recovered constructively by combining the vertex maps of the cube and edge-oriented tetrahedron.
Based on this, we consider now the birational group obtained by combining the generators of the $n$-cube birational group (\ref{cubegenerators}) and edge-oriented $n$-simplex birational group (\ref{simplexgenerators}) for a generic value of $n$.
It is convenient to establish first the following:
\begin{lemma}\label{PIJLEM}
Let $\sigma_1,\ldots,\sigma_n$ and $\omega$ be as defined in Section \ref{fg} in terms of an underlying rational triplet-pair system which is 5-simplex consistent (Definition \ref{fsc}).
Then
\begin{equation}
(\omega\sigma_i\sigma_j)^3=\pi_{i,j}, \quad i\neq j,\label{pijlem}
\end{equation}
where $\pi_{i,j}$ is the permutation mapping on the $2\times n$ array that interchanges row $i$ and $j$.
\end{lemma}
\begin{proof}
The general case described follows by symmetry from the particular case $i=1$, $j=2$ and $n=3$ that was established in the proof of Proposition \ref{coincidence}.
The action of the three participating mappings $\omega$, $\sigma_1$ and $\sigma_2$ on row $k$ for $k>3$ is similar to their action on row $3$.
But in the case $n=3$ the composed mapping $(\omega\sigma_1\sigma_2)^3$ fixes row $3$, and therefore for generic $n$ this mapping also fixes all rows $k$ for $k>3$.
Thus (\ref{pijlem}) holds in the case $i=1$, $j=2$ with $n$ arbitrary.
From the definition (\ref{vertex}), permutation of rows $1,\ldots,n$ corresponds to permutation of mappings $\sigma_1,\ldots,\sigma_n$, and application of such permutations leads directly to (\ref{pijlem}) as a consequence of the case $i=1$, $j=2$.
\end{proof}
This implies immediately the following.
\begin{prop}\label{cvmg}
If the rational triplet-pair system underlying construction of rational mappings in Section \ref{fg} is 5-simplex consistent, then the birational group obtained from the combined generators of the $n$-cube (\ref{cubegenerators}) and $n$-simplex (\ref{simplexgenerators}), i.e., the complete vertex map group, is equal to the group generated by $\sigma_1,\ldots,\sigma_n,\omega$.
\end{prop}
The remaining question we answer in this section is how to characterise this complete group.
In fact we will assume the integer $n$ is fixed and consider the group generated by 
\begin{equation}
\sigma_1,\ldots,\sigma_m,\omega,\label{gens}
\end{equation}
for $m<n$.
This has the advantage that if $m'<m$, then $\langle\sigma_1,\ldots,\sigma_{m'},\omega\rangle$ is a subgroup of $\langle\sigma_1,\ldots,\sigma_{m},\omega\rangle$.
The following relations are key.
\begin{equation}
\begin{split}
\sigma_i^2=\omega^2=(\sigma_i\omega)^4 = \id,\\
(\sigma_i\sigma_j)^2 = (\sigma_i\omega\sigma_j\omega)^3 = \id, & \qquad |\{i,j\}|=2,\\
((\sigma_i\sigma_j\omega)^2\sigma_k\omega)^2 = \id, & \qquad |\{i,j,k\}|=3.
\end{split}
\label{E6rels}
\end{equation}
These relations encode the stronger 5-simplex consistency property:
\begin{prop}\label{cgeqv}
Suppose $2<m<n$.
The underlying rational triplet-pair system is 5-simplex consistent (Definition \ref{fsc}), if and only if the mappings $\sigma_1,\ldots,\sigma_m$ and $\omega$, constructed from it in Section \ref{fg}, satisfy the relations (\ref{E6rels}).
\end{prop}
\begin{proof}
By construction, the group generated by mappings (\ref{fivegenerators}) allows variables on all edges of the 5-simplex to be obtained consistently as rational functions of those from the initial data array (\ref{gg}) (cf. Definition \ref{fsc}), provided they satisfy the $A_5$ group relations.
Thus to prove the 5-simplex consistency it suffices to show that the $A_5$ Coxeter relations hold amongst the mappings (\ref{fivegenerators}) in the particular case $n=3$, as a consequence of the relations (\ref{E6rels}).
This is straightforward to verify, the Coxeter relations fall into three kinds.
The first kind are $(\sigma_i\sigma_j)^2=\id$ and this relation appears directly in (\ref{E6rels}), the second kind are $(\sigma_i\sigma_i^\omega)^2=\id$ which appears in (\ref{E6rels}) written in the form $(\sigma_i\omega)^4=\id$, and the third kind are $(\sigma_i\sigma_j^\omega)^3=\id$ with $i\neq j$, which appears in (\ref{E6rels}) in the form $(\sigma_i\omega\sigma_j\omega)^3=\id$.

The subtlety in the converse statement is that it relates to the generic case of $m<n$, where $n$ is arbitrary.
That $\sigma_i^2=\id$ follows from the assumed permutation symmetry of the rational triplet-pair system.
That $\omega^2=\id$ is obvious. 
And assuming these first two relations, the third relation $(\sigma_i\omega)^4=\id$ expresses the commutativity of $\sigma_i$ and $\sigma_i^\omega$, which is also a consequence of the assumed symmetry (established before in (\ref{comm})).
Thus it remains to establish the remaining three relations in (\ref{E6rels}) as a consequence of the 5-simplex consistency.
The two relations $(\sigma_i\sigma_j)^2=\id$ and $(\sigma_i\omega\sigma_j\omega)^3=\id$ for generic $i\neq j$ follow from the particular case $i=1$, $j=2$ and $n=3$ by using the same argument used in the proof of Lemma \ref{PIJLEM}.
As mentioned in the first part of this proof, the particular case $i=1$, $j=2$ and $n=3$ of these two relations are present in the $A_5$ relations between generators (\ref{fivegenerators}) and are therefore a consequence of the 5-simplex consistency.
The final relation in (\ref{E6rels}) expresses the obvious commutativity between the row-permutation mapping $\pi_{i,j}$ and the mapping $\sigma_k$ for $k\not\in\{i,j\}$, it is only obscured by the fact that $\pi_{i,j}$ is written in the form established in Lemma \ref{PIJLEM}.
To see this we write the following chain 
\begin{equation}
\begin{split}
(\pi_{i,j}\sigma_k)^2 & =\omega\sigma_i\sigma_j\omega\sigma_i\sigma_j\omega\sigma_i\sigma_j\sigma_k\pi_{i,j}\sigma_k\\
& = \omega\sigma_i\sigma_j\omega\sigma_i\sigma_j\omega\sigma_k\pi_{i,j}\sigma_j\sigma_i\sigma_k\\
& = \omega((\sigma_i\sigma_j\omega)^2\sigma_k\omega)^2\omega,
\end{split}
\end{equation}
in which we have used the formula (\ref{pijlem}), the fact that $\pi_{i,j}\sigma_i=\sigma_j\pi_{i,j}$, and that $\sigma_k$ commutes with $\sigma_i$ and $\sigma_j$ which follows from the already established relations.
\end{proof}

Thus we have characterised the abstract structure of the group of Proposition \ref{cvmg} by establishing a finite presentation for it.
We now identify that, like the subgroups participating in its construction, the complete vertex map group is also a Coxeter group.
\begin{figure}[t]
\begin{center}
\begin{tikzpicture}[thick]
  \tikzstyle{every node}=[draw,circle,fill=black,minimum size=5pt,inner sep=0pt,label distance=2pt]
  \draw[dashed] (3,0.5) -- (5,0.5);
  \draw (0,0) node[label=above:$2$]{} -- (1,0) node[label=above:$4$]{} -- (2,0.5) -- (3,0.5) node[label=above:$6$]{};
  \draw (0,1) node[label=above:$1$]{} -- (1,1) node[label=above:$3$]{} -- (2,0.5) node[label=above:$5$]{};
  \tikzstyle{every node}=[draw,circle,fill=black,minimum size=5pt,inner sep=0pt,label distance=-5pt]
  \draw (5,0.5) -- (5,0.5) node[label=above:$m\!+\!1$]{} -- (6,0.5) node[label=above:$m\!+\!2$]{};
\end{tikzpicture}
\end{center}
\caption{Coxeter-Dynkin diagram associated with the complete vertex-map group.}
\label{cdd}
\end{figure}
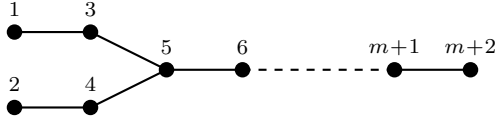
\begin{prop}\label{vmg}
Suppose $m>1$.
Then the finitely presented group with generators (\ref{gens}) and relations (\ref{E6rels}) is isomorphic to the Coxeter group of Figure \ref{cdd} extended by its graph automorphism.
\end{prop}
\begin{proof}
The target group is, in precise terms, the finitely presented group with generators $t_0,\ldots,t_{m+2}$ and relations as follows.
The first set of relations, not involving generator $t_0$, are encoded in the Coxeter graph of Figure \ref{cdd}, where labels on nodes of the graph correspond to indices of the generators $t_1,\ldots,t_{m+2}$.
The remaining relations all involve the generator $t_0$, and are as follows:
\begin{equation}
\begin{split}
t_0^2 = \id,&\\
t_0 t_1 = t_2 t_0,&\\
t_0 t_3 = t_4 t_0,&\\
t_0 t_i = t_i t_0,& \qquad i \in \{5,\ldots,m+2\}.\\
\end{split}\label{ecr}
\end{equation}
Thus the action of $t_0$ can be identified with the non-trivial permutation of nodes under which the graph is invariant.

We proceed by exhibiting an isomorphism between this target group and the finitely presented group defined by generators (\ref{gens}) and relations (\ref{E6rels}):
\begin{equation}
\begin{array}{l}
t_0 = \sigma_1,\\
t_1 = \omega,\\
t_2 = \sigma_1\omega\sigma_1,\\
t_3 = \omega\sigma_1\sigma_2\omega\sigma_2\sigma_1\omega,\\
t_4 = (\omega\sigma_1\sigma_2)^3,\\
t_5 = (\omega\sigma_2\sigma_3)^3,\\
\quad \vdots\\
t_{m+2} = (\omega\sigma_{m-1}\sigma_{m})^3,
\end{array}
\qquad
\begin{array}{l}
\omega = t_1,\\
\sigma_1 = t_0,\\
\sigma_2 = t_0^{t_4},\\
\sigma_3 = t_0^{t_5t_4},\\
\sigma_4 = t_0^{t_6t_5t_4},\\
\sigma_5 = t_0^{t_7t_6t_5t_4},\\
\quad \vdots\\
\sigma_m = t_0^{t_{m+2}\cdots t_5t_4}.
\end{array}
\label{iso}
\end{equation}
Again here the notation $f^g$ denotes conjugation of $f$ by $g$, $f^g=gfg^{-1}$.
The precise claim that (\ref{iso}) defines the desired isomorphism is in two parts.
First, the elements $t_0,\ldots,t_{m+2}$ defined by the expressions on the left in (\ref{iso}) satisfy the relations encoded in the Coxeter group of Figure \ref{cdd} and its graph automorphism (\ref{ecr}), as well as the equations on the right in (\ref{iso}), by virtue of relations (\ref{E6rels}).
Second, the elements $\omega,\sigma_1,\ldots,\sigma_m$ defined by the expressions on the right in (\ref{iso}) satisfy the relations (\ref{E6rels}), as well as equations on the left in (\ref{iso}), by virtue of the relations on $t_0,\ldots,t_{m+2}$ that are encoded in the Coxeter graph of Figure \ref{cdd} and its graph automorphism (\ref{ecr}).
Both statements can be verified by straightforward but numerous calculations.

The explicit isomorphism (\ref{iso}) was obtained by generalising an isomorphism for the case $m=4$ found using the computer algebra system Magma \cite{WCB}.
\end{proof}
\begin{table}
\begin{center}
\begin{tabular}{lll}
Group & Order & Isomorphic to \\
\hline
$\langle\sigma_1,\sigma_2,\omega\rangle$ & $2\times 36$ & $A_1 \ltimes (A_2 \times A_2)$ \\
$\langle\sigma_1,\sigma_2,\sigma_3,\omega\rangle$ & $2\times 720$ & $A_1 \ltimes A_5$ \\
$\langle\sigma_1,\sigma_2,\sigma_3,\sigma_4,\omega\rangle$ & $2\times 51840$ & $A_1 \ltimes E_6$ \\
$\langle\sigma_1,\sigma_2,\sigma_3,\sigma_4,\sigma_5,\omega\rangle$ & $2\times \infty$ & $A_1 \ltimes \wt{E}_6$ 
\end{tabular}
\end{center}
\caption{
Sequence of subgroups of the complete vertex-map group.
}
\label{vmt}
\end{table}

The abstract structure of the vertex map group encodes a generalised lattice geometry.
Specifically, underlying the sequence of groups is a sequence of triplet-pair arrangements (Definition \ref{tp}), and each element of the group corresponds to a symmetry of the associated arrangement (Definition \ref{symdef}).
For each value of $m$ this arrangement contains, by construction, the triplet-pair arrangements of both the $m$-cube and the edge-oriented $m$-simplex.
Within Figure \ref{cdd} one can also see the group of a single triplet-pair ($m=2$) and the 5-simplex arrangement ($m=3$), cf. Proposition \ref{coincidence}.
Continuing, the combined vertex map group decomposes as in Table \ref{vmt}.
The fact that the arrangement becomes infinite beyond the case $m=4$ represents a significant and un-anticipated depth beyond the kind of lattice geometry of the $m$-cube or edge-oriented $m$-simplex.
A natural step to understanding the dynamics in the ambient space as a whole, is to relate the lattice geometry to the group.

\section{The generalised lattice geometry}\label{LG}
In the previous section we characterised the abstract structure of the complete vertex-map group, this group encodes the full ambient domain on which it is natural to consider a system with the 5-simplex consistency.
In this section we will explain how the geometry of the domain may be decoded from the group, which can be viewed as the culmination of efforts from Sections \ref{5S} and \ref{FH}, completing the unification of the $n$-cube and edge-oriented $n$-simplex lattice geometries.
We then make contact with the combinatorics of the 27 lines of a cubic surface in the case $n=4$, and discuss the generalised lattice geometry in the context of integrable dynamics.

\subsection{Decoding the lattice geometry from the group}\label{lgp}
The extended domain itself is a triplet-pair arrangement (Definition \ref{tp}).
To achieve the coordinatisation, we exploit the fact that variables of the arrangement are in natural bijection with a particular conjugacy class of its symmetry group (Definition \ref{symdef}).
This allows coordinatisation of the arrangement using nothing more than the group itself.
\begin{definition}\label{Gtp}
Denote by $G$ the finitely presented group defined by generators (\ref{gens}) and relations (\ref{E6rels}) (cf. Proposition \ref{vmg}).
The triplet-pair arrangement associated with $G$ is a set of variables, assigned to elements of the conjugacy class of $\sigma_1\in G$,
\begin{equation}
x(\sigma_1^g), \quad g\in G,
\end{equation}
arranged into the triplet-pairs
\begin{equation}
x(\sigma_1^{g}),x(\sigma_2^{g\omega}),x(\sigma_2^{g\sigma_1\omega}), \quad
x(\sigma_2^{g}),x(\sigma_1^{g\omega}),x(\sigma_1^{g\sigma_2\omega}), \quad
g\in G.\label{gc}
\end{equation}
\end{definition}
This definition is consistent because $\sigma_2$ is conjugate to $\sigma_1$ in $G$, specifically $\sigma_2=\sigma_1^{\omega\sigma_1\omega\sigma_2\omega}$, which is a re-writing of the case $i=2$, $j=1$ of the fifth relation in (\ref{E6rels}) using the first two relations.

The main combinatorial features of this constructed domain are as follows.
\begin{prop} \label{latext}
The triplet-pair arrangement of Definition \ref{Gtp} is symmetry-complete and contains as sub-arrangements both the $m$-cube and edge-oriented $m$-simplex arrangements of Definitions \ref{CTP} and \ref{EOTTP}.
In the case $m=3$ it coincides with the 5-simplex arrangement of Definition \ref{5SPC}.
\end{prop}
\begin{proof}
Observe that the symmetry group of the triplet-pair (\ref{gc}) associated with some $g\in G$ can be identified with the inner automorphisms of the group $\langle \sigma_1^g, \sigma_2^g, \omega^g\rangle$.
The action of this subgroup extends naturally to the whole of $G$, it clearly permutes elements of any conjugacy class, and thus permutes both variables and triplet-pairs, and this is sufficient for the symmetry-completeness.

The $m$-cube, edge-oriented $m$-simplex and the 5-simplex sub-arrangements are obtained by restricting to a particular subgroup $H<G$, which is, respectively, as follows
\begin{eqnarray}
&\langle\sigma_1,(\omega\sigma_1\sigma_2)^3,\ldots,(\omega\sigma_{m-1}\sigma_m)^3\rangle, \label{G1}\\
&\langle\sigma_1\omega\sigma_1,(\omega\sigma_1\sigma_2)^3,\ldots,(\omega\sigma_{m-1}\sigma_m)^3,\omega\rangle,\label{G2}\\
&\langle\sigma_1,\sigma_3^\omega,\sigma_2,\sigma_1^\omega,\sigma_3\rangle.\label{G3}
\end{eqnarray}
The corresponding triplet-pair sub-arrangement involves the subset of variables $\{x(\sigma_1^g),x(\sigma_1^{g\omega}):g\in H\}$ and the subset of triplet-pairs obtained by letting $g$ in (\ref{gc}) run over the subgroup $H$ rather than the whole of $G$.
Subgroup (\ref{G1}) is isomorphic to $BC_m$, subgroup (\ref{G2}) is isomorphic to $A_m\times A_1$, and subgroup (\ref{G3}) is isomorphic to $A_5$.
The generators listed correspond to standard generators of these groups, and the characterising relations between the generators can be verified as a consequence of relations (\ref{E6rels}).
Alternatively, characterisation of these subgroups of $G$ is clear by construction on the level of the birational representations in Propositions \ref{BRBC}, \ref{BRA} and \ref{BRA5} (modulo Lemma \ref{PIJLEM}).
We will now proceed by giving explicit identifications between coordinates for these respective sub-arrangements.

The notation 
\[
\sigma_I = \prod_{i\in I}\sigma_i, \quad I\subseteq \{1,\ldots,m\}
\]
(ordering is not important because $\sigma_i$ and $\sigma_j$ commute), allows to write an association between variables (\ref{ecvars}) of the $m$-cube arrangement (Definition \ref{CTP}) and the restriction via subgroup (\ref{G1}) of the arrangement of Definition \ref{Gtp} as follows:
\begin{equation}\label{cubeidents}
x(\sigma_i), \quad x(\sigma_i^{\sigma_I\omega}) \quad \leftrightarrow \quad y_i, \quad x_i^I
\end{equation}
where $i\in\{1,\ldots,m\}$ and $I \subseteq \{1,\ldots,m\} \setminus \{i\}$.
It is straightforward to verify via this association that the action of the generators in (\ref{G1}) correspond, respectively, to reflection through axis 1 of the $m$-cube, and then reflections that interchange axis $i$ and $i+1$ for $i\in\{1,\ldots,m\}$ whilst leaving the remaining axes unchanged.
These are standard generators of the $m$-cube symmetry group.
The explicit identification (\ref{cubeidents}) makes it straightforward to verify that the triplet-pairs of the two arrangements correspond.

The similar association between variables (\ref{evars}) of the edge-oriented $m$-simplex arrangement (Definition \ref{EOTTP}) and the restriction via subgroup (\ref{G2}) of the arrangement of Definition \ref{Gtp} is as follows:
\begin{equation}\label{simplexidents}
\begin{array}{ll}
x(\sigma_i^\omega), & x(\sigma_i)\\
x(\sigma_i^{\omega\sigma_j\omega}), & x(\sigma_i^{\sigma_j\omega})
\end{array}
\quad
\leftrightarrow
\quad
\begin{array}{ll}
x_{1(i+1)}, &x_{(i+1)1}\\
x_{(j+1)(i+1)},&x_{(i+1)(j+1)}
\end{array}
\end{equation}
where $i,j\in\{1,\ldots,m\}$ and $i\neq j$.
Recall that the indices of variables (\ref{evars}) correspond to vertices of the $m$-simplex.
This association therefore allows to verify that action of the generators in (\ref{G2}) correspond, respectively, to interchange of the $m$-simplex vertices $1\leftrightarrow 2$, $2\leftrightarrow 3$, $\ldots$, $m\leftrightarrow (m+1)$.
I.e., to the standard generators of the $m$-simplex symmetry group.
The final generator corresponds to orientation reversal of all edges.
Again, given the explicit identification (\ref{simplexidents}), it is straightforward to verify that the triplet-pairs of the two arrangements correspond.

Finally, we give an identification similar to (\ref{cfs}) and (\ref{otfs}).
\begin{equation}
\begin{array}{lll}
\sigma_1 \ & \sigma_2 \ & \sigma_3  \\[0.3em]
\sigma_1^\omega \ & \sigma_2^\omega \ & \sigma_3^\omega  \\[0.3em]
\sigma_1^{\sigma_2\omega} \ & \sigma_2^{\sigma_3\omega} \ & \sigma_3^{\sigma_1\omega}  \\[0.3em]
\sigma_1^{\sigma_3\omega} \ & \sigma_2^{\sigma_1\omega} \ & \sigma_3^{\sigma_2\omega}  \\[0.3em]
\sigma_1^{\sigma_2\sigma_3\omega} \ & \sigma_2^{\sigma_3\sigma_1\omega} \ & \sigma_3^{\sigma_1\sigma_2\omega}
\end{array}
\leftrightarrow
\qquad
\begin{array}{lll}
w_{12} \ & w_{34} \ & w_{56} \\[0.3em]
w_{45} \ & w_{16} \ & w_{23} \\[0.3em]
w_{35} \ & w_{15} \ & w_{13} \\[0.3em]
w_{46} \ & w_{26} \ & w_{24} \\[0.3em]
w_{36} \ & w_{25} \ & w_{14}
\end{array}
\label{cva}
\end{equation}
This is an explicit identification between elements of the conjugacy class of $\sigma_1$ in the group $\langle \sigma_1,\sigma_2,\sigma_3,\omega\rangle$, and the variables of the 5-simplex triplet-pair arrangement (Definition \ref{5SPC}).
It allows to directly verify the final claim of the proposition.
\end{proof}
Notice that Proposition \ref{latext} is purely combinatorial, as is the proof.
Thus it could have been stated and proven earlier at the end of Section \ref{5S}, but it would have meant introducing the finitely presented group $G$ (of Definition \ref{Gtp}) without motivation.
Unlike in Proposition \ref{5s}, we have not proven the minimality of the symmetry-complete extension, though it is likely to be minimal as we will argue later.

The natural initial-value problem for the 5-simplex consistent systems on the extended domain is as follows.
\begin{prop}\label{latgeom}
Impose a 5-simplex consistent rational triplet-pair system (Definitions \ref{tps} and \ref{fsc}) on all triplet-pairs of the arrangement of Definition \ref{Gtp}.
Then all remaining variables are composed rational functions of the unconstrained subset
\begin{equation}
x(\sigma_1^\omega),\ldots,x(\sigma_m^\omega),x(\sigma_1),\ldots,x(\sigma_m).\label{fvars}
\end{equation}
\end{prop}
\begin{proof}
This is constructive based on the procedure of birational reformulation described in Section \ref{BRF}.
The natural action of $g\in G$ on the defined variables, specifically $g(x(\nu))=x(\nu^g)$, clearly gives a permutation symmetry of the specified system of equations.
The central role is then played by the following function on elements $g\in G$:
\begin{equation}
X(g) := \flm{x(\sigma_1^{g \omega}),x(\sigma_1^g)}{x(\sigma_m^{g \omega}),x(\sigma_m^g)}.
\end{equation}
Here we approach the system from the converse direction to Section \ref{BRF}. 
Specifically, we wish to show that, as a consequence of the imposed equations, all variables can be obtained as composed rational functions of the variables in $X(\id)$, and that the variables in $X(\id)$ are themselves unconstrained.
Due to the imposed system on triplet-pairs (\ref{gc}), the action by composition from the left of some generator $h\in\{\sigma_1,\ldots,\sigma_m\}$ on some fixed $g\in G$ induces a rational action on $X(g)$ as follows:
\begin{equation}
X(hg) = h(X(g)).\label{induce}
\end{equation}
This defines the rational action of $h\in\{\sigma_1,\ldots,\sigma_m\}$ on the array of variables $X(g)$.
Due to the way $X$ has been defined, this action coincides with the previously defined rational action (\ref{vertex}), which is easily checked.
Similarly, the action of $\omega$ induced by (\ref{induce}) coincides with the action defined previously (\ref{omegadef}).

In Proposition \ref{cgeqv}, we showed that the 5-simplex consistency implies that generators (\ref{vertex}) and (\ref{omegadef}) satisfy relations (\ref{E6rels}), and therefore that (\ref{induce}) holds for all $h,g\in G$.
Thus equations
\begin{equation}
X(g) = g(X(\id)), \quad g \in G, \label{actdef}
\end{equation}
allow to obtain variables present in the array $X(g)$ in terms of those present in $X(\id)$, i.e., in terms of the subset of variables (\ref{fvars}).
Conversely, if all variables are obtained in terms of $X(\id)$ by (\ref{actdef}), then all equations on triplet-pairs (\ref{gc}) are satisfied, because they are equivalent to the equations $X(\sigma_1 g)=\sigma_1(X(g)),\ g\in G$, which, due to the consistency (\ref{induce}), are clearly a consequence of (\ref{actdef}).
\end{proof}
Proposition \ref{latgeom} concludes our construction of the natural lattice and initial-value problem for the 5-simplex consistent systems.
We now make several remarks on this result.

\begin{remark}
We have used the term {\it composed rational functions} in Proposition \ref{latgeom} because if some element $g\in G$ cannot be expressed as a finite word in the generators, then the corresponding variable $x(\sigma_1^g)$ is obtained in terms of the variables (\ref{fvars}) through composition of an infinite sequence of rational functions, so it is possible that such dependence is not itself rational.
\end{remark}

\begin{remark}
Although we did not prove that the symmetry-complete extension of the $m$-cube and edge-oriented $m$-simplex triplet-pair arrangements given in Definition \ref{Gtp} is the smallest possible (cf. Proposition \ref{latext}), what can be observed is that the number of initial data is the same for the extended domain as it is for the $m$-cube and edge-oriented $m$-simplex domains which it contains ($2m$).
This suggests the minimality of the extension, but it does not constitute a satisfactory proof.
The combinatorial proof of minimality used in the case $m=3$ (Proposition \ref{5s}) does not apply directly in the case where the extension has infinite extent.
\end{remark}

\begin{remark}
Due to the fact that rational triplet-pair systems (Definition \ref{tps}) are necessarily quadrirational, it is likely that the systems in Proposition \ref{conform} exhaust the examples of {\it scalar} 5-simplex consistent systems.
However, we have taken care to show how the natural lattice geometry emerges exclusively from the 5-simplex consistency property, and therefore the lattice is natural also for potential generalisations, which could be multi-component, matrix or non-commutative triplet-pair systems.
\end{remark}

\subsection{Case $m=4$}
In the case $m=4$ of Definition \ref{latext}, the triplet-pair arrangement is still finite (see Table \ref{vmt}).
This lattice shares the same beautiful combinatorics as the configuration of the 27 lines of a cubic surface.

Amongst the substantial literature on this subject, we cite here mainly the work of Baker \cite{Baker} who studied the geometry of Burkhardt's quartic \cite{HB}, which is a geometric object related to the cubic surface. 
This study was purely in terms of position in four-dimensional space of the 45 singular points, called nodes, which he found were in natural correspondence with 45 particular elements of the symmetry group under which the nodes are permuted, called projections.
The projections were proven to form a complete conjugacy class by Todd \cite{Todd}, who also makes a concise summary of Bakers extensive geometric description of the 45 nodes.
The set of 45 projections are in correspondence with the conjugacy class to which we have associated variables in case $m=4$ of Definition \ref{latext}.

We recount a part of Bakers geometric description, in order to identify in that setting what we have called triplet-pairs here.
The 45 nodes lie by threes on lines, called $\kappa$-lines, with 16 lines through each node and 240 lines in total.
The nodes and $\kappa$-lines therefore form a $(45_{16}240_{3})$ point-line configuration.
The set of all $\kappa$-lines intersecting a particular one contain every node of the configuration except three.
Those three remaining nodes themselves lie on a second $\kappa$-line, which is called polar to the first.
The relationship is reciprocal, and the 120 pairs of polar $\kappa$-lines can be identified with the triplet-pairs in (\ref{gc}).

The $m=4$ triplet-pair arrangement includes 36 $m=3$ (5-simplex) triplet-pair arrangements.
This sub-arrangement is left unnamed by Baker, but is called a 15-set by Todd.
On the other hand, one of Bakers additional contributions was to provide a coordinatisation, and this implicitly acknowledges the 15-set because it directly extends the natural coordinatisation of it given earlier in Definition \ref{5SPC}.
Bakers coordinatisation allows to write down the $m=4$ triplet-pair arrangement as follows.
\begin{prop}[notation due to Baker \cite{Baker}]
The triplet-pair arrangement of Definition \ref{Gtp} in the case $m=4$ consists of 45 variables arranged into 120 triplet-pairs.
If we denote 15 of the variables by
\begin{equation}
w_{ij}=w_{ji}, \quad i,j\in\{1,2,3,4,5,6\}, |\{i,j\}|=2,
\end{equation}
and 30 of them by
\begin{equation}
w_{(ij,kl,mn)}=w_{(ji,kl,mn)}=w_{(kl,mn,ij)}, \quad \{i,j,k,l,m,n\}=\{1,2,3,4,5,6\},
\end{equation}
then the 120 triplet-pairs may be written as follows: there are 10 of the form
\begin{equation}
w_{ij},w_{jk},w_{ki}, \quad w_{lm},w_{mn},w_{nl},\label{a120}
\end{equation}
90 of the form
\begin{equation}
w_{ij},w_{(ik,jl,mn)},w_{(jk,il,mn)}, \quad w_{mn},w_{(mk,nl,ij)},w_{(nk,ml,ij)},\label{b120}
\end{equation}
and 20 of the form
\begin{equation}
w_{(ij,kl,mn)},w_{(il,kn,mj)},w_{(in,kj,ml)}, \quad w_{(ij,mn,kl)},w_{(il,mj,kn)},w_{(in,ml,kj)},\label{c120}
\end{equation}
where in (\ref{a120}), (\ref{b120}) and (\ref{c120}), $\{i,j,k,l,m,n\}=\{1,2,3,4,5,6\}$.
The subset of 8 unconstrained variables appearing in case $m=4$ of Proposition \ref{latgeom} can be taken as
\begin{equation}
w_{12}, w_{34}, w_{56}, w_{(12,34,56)}, w_{45}, w_{16}, w_{23}, w_{(16,45,23)}.
\end{equation}
\end{prop}
\begin{proof}
We proceed by giving an extension of the association (\ref{cva}).
\begin{equation*}
\begin{array}{lll}
\sigma_1^{\sigma_4\omega} \ & \sigma_2^{\sigma_4\omega} \ & \sigma_3^{\sigma_4\omega} \\[0.3em]
\sigma_1^{\sigma_4\sigma_2\omega} \ & \sigma_2^{\sigma_4\sigma_3\omega} \ & \sigma_3^{\sigma_4\sigma_1\omega}  \\[0.3em]
\sigma_1^{\sigma_4\sigma_3\omega} \ & \sigma_2^{\sigma_4\sigma_1\omega} \ & \sigma_3^{\sigma_4\sigma_2\omega}  \\[0.3em]
\sigma_1^{\sigma_4\sigma_2\sigma_3\omega} \ & \sigma_2^{\sigma_4\sigma_3\sigma_1\omega} \ & \sigma_3^{\sigma_4\sigma_1\sigma_2\omega} \\[0.3em]
\sigma_4^{\sigma_1\omega} \ & \sigma_4^{\sigma_2\omega} \ & \sigma_4^{\sigma_3\omega} \\[0.3em]
\sigma_4^{\sigma_2\sigma_3\omega} \ & \sigma_4^{\sigma_3\sigma_1\omega} \ & \sigma_4^{\sigma_1\sigma_2\omega} \\[0.3em]
\sigma_1^{\omega\sigma_2\sigma_3\sigma_4\omega} \ & \sigma_1^{\sigma_1\omega\sigma_2\sigma_3\sigma_4\omega} \ & \sigma_4 \\[0.3em]
\sigma_2^{\omega\sigma_1\sigma_3\sigma_4\omega} \ & \sigma_2^{\sigma_1\omega\sigma_1\sigma_3\sigma_4\omega} \ & \sigma_4^{\omega} \\[0.3em]
\sigma_3^{\omega\sigma_1\sigma_2\sigma_4\omega} \ & \sigma_3^{\sigma_1\omega\sigma_1\sigma_2\sigma_4\omega} \ & \sigma_4^{\sigma_1\sigma_2\sigma_3\omega} \\[0.3em]
\sigma_4^{\omega\sigma_1\sigma_2\sigma_3\omega} \ & \sigma_4^{\sigma_1\omega\sigma_1\sigma_2\sigma_3\omega} \ & \sigma_4^{\omega\sigma_4\omega\sigma_1\sigma_2\sigma_3\omega}
\end{array}
\leftrightarrow
\quad
\begin{array}{lll}
w_{(12,35,46)} \ & w_{(15,26,34)} \ & w_{(13,24,56)} \\[0.3em]
w_{(12,45,36)} \ & w_{(16,25,34)} \ & w_{(14,56,23)} \\[0.3em]
w_{(12,36,45)} \ & w_{(16,34,25)} \ & w_{(14,23,56)} \\[0.3em]
w_{(12,46,35)} \ & w_{(15,34,26)} \ & w_{(13,56,24)} \\[0.3em]
w_{(13,26,45)} \ & w_{(16,35,24)} \ & w_{(15,46,23)} \\[0.3em]
w_{(15,36,24)} \ & w_{(13,25,46)} \ & w_{(14,26,35)} \\[0.3em]
w_{(13,45,26)} \ & w_{(16,23,45)} \ & w_{(12,34,56)} \\[0.3em]
w_{(16,24,35)} \ & w_{(14,35,26)} \ & w_{(16,45,23)} \\[0.3em]
w_{(15,23,46)} \ & w_{(13,46,25)} \ & w_{(14,25,36)} \\[0.3em]
w_{(14,36,25)} \ & w_{(15,24,36)} \ & w_{(12,56,34)}
\end{array}
\end{equation*}
Combined with (\ref{cva}), this defines a correspondence between the 45 elements of the conjugacy class of $\sigma_1$ in the group $\langle \sigma_1,\sigma_2,\sigma_3,\sigma_4,\omega\rangle$ and the 45 variables in the given notation. 
That the given list exhausts the conjugacy class can be checked directly because the group in question is finite.
From the explicit correspondence it can be checked that the triplet-pairs appearing in case $m=4$ of Proposition \ref{latgeom}, i.e., 
\begin{equation}
\sigma_1^{g},\sigma_2^{g\omega},\sigma_2^{g\sigma_1\omega}, \qquad
\sigma_2^{g},\sigma_1^{g\omega},\sigma_1^{g\sigma_2\omega}, \qquad g\in \langle \sigma_1,\sigma_2,\sigma_3,\sigma_4,\omega\rangle,
\end{equation}
are in one-to-one correspondence with the triplet-pairs (\ref{a120}), (\ref{b120}) and (\ref{c120}).
\end{proof}

Adler, Bobenko and Suris, when introducing the scalar quadrirational models \cite{ABSf}, gave also a geometric incidence picture in terms of pencils of conics.
The refinement of this geometric picture, accounting for the additional symmetry and consistency, is potentially very interesting. 
In particular, the exceptional nature of the associated combinatorics indicates a likelihood that such refinement will connect with the cubic surface.

\subsection{Subgroup dynamics and integrability}\label{SDI}
Here we discuss the familiar sense of integrability that is recovered when the restriction is made to particular subgroups of the complete vertex-map group, $G$ in Definition \ref{Gtp}.

By construction, the complete vertex-map group contains $A_n$ and $BC_n$ subgroups corresponding to the edge-oriented $n$-simplex and $n$-cube domains (see Proposition \ref{latext}).
In both of these domains it is envisaged that $n$ is arbitrarily large, and that the number of degrees of freedom is therefore un-bounded.
Integrable dynamics associated with the $n$-cube is on a quad-graph and is of KdV-type, whereas the edge-oriented $n$-simplex (as discussed in Section \ref{OTC}) leads to a simpler kind of integrability on a triangle-graph.

Due to the nature of the complete vertex map group, the unbounded domain required for non-trivial dynamics is also present for finite $n$, the simplest case being $n=5$.
The group in this case is the affine reflection group $\wt{E_6}$, such affine groups are familiar from dynamics in the Painlev\'e setting, though the birational representation here looks very different.
Nevertheless it would be natural to adopt the similar point of view, that dynamics are associated with the normal abelian subgroup, whilst the action of the finite quotient group ($E_6$) corresponds to symmetries of the dynamical system.
This would give a natural notion of integrability, but it remains open whether the systems here can actually be integrated in terms of (discrete) Painlev\'e transcendents or their generalisations.

Restriction to subgroup dynamics by taking a reduced set of generators leads to an additional variety of possibilities.
An interesting case would be the group
\begin{equation}
\langle t_1, t_3,\ldots,t_{10}\rangle,
\end{equation}
written in terms of elements given in (\ref{iso}), which constitutes a birational representation of $\wt{E_8}$.

The lattice in its full generality is a combinatorial structure which of itself carries no obvious notion of distance, it is {\it non-metric}.
In this setting the identification of subgroups in which there is a meaningful notion of distance (like in the aforementioned three cases) is of itself useful, because it implies non-trivial dynamics. 
For such dynamics there is the inherent structure of consistent embedding in the ambient space, and it is natural to expect this leads to integrability in other senses. 
This is something which could potentially be tested by the notion of algebraic entropy \cite{FalVia,bv} for instance.

\section{Discussion}\label{D}
It has been established that there is a stronger consistency property, beyond the usual integrability, of quad-graph systems at the top level of broader classes defined in \cite{ABSf} and \cite{AtkNie}, namely the quadrirational maps and the multi-quadratic quad-equations.
The key feature distinguishing the top-level systems is additional permutation symmetry that puts parameters on an equal footing with the variables.
The stronger consistency, which has combinatorics of the 5-simplex, can be viewed as an extension of the usual cubic consistency that respects this symmetry.
That is, there is no distinction between parameters and variables also on the level of the consistency property.
This aspect of the consistency theory is purely combinatorial, and has been formalised here in the notion of {\it symmetry completeness} (Definition \ref{sc}).

The usual situation due to the cubic consistency, is that a solution of the planar quad-graph system can be extended to the hypercube in which the quad-graph is embedded.
The significance of the 5-simplex consistency is that it leads, in a similar fashion, to a yet larger embedding lattice which extends beyond the hypercube.
The characterisation the extended lattice is one of the main technical problems tackled in this paper.
The key method used is the reformulation of the system as a birational representation of the symmetry group of the lattice, which turns out to be the Coxeter group given at the end of Section \ref{FH} (Figure \ref{cdd}).

The significance of the extended lattice is that it constitutes a unification between distinct classes of integrable systems (cf. Section \ref{SDI}).
Besides the original quad-graph equations, we identify two other classes of systems that are contained by restriction to a sublattice.
In the first case there is a system defined on a triangle-graph, which is integrable because solutions can be extended consistently to an embedding multidimensional simplex.
In the second case, which is most unanticipated, there are finite degree-of-freedom rational mappings that can be considered integrable because they are identified as elements of a birational affine reflection group.
The extended lattice therefore sits above integrable systems from three a-priori quite separate settings, unifying their respective notions of integrability in the 5-simplex consistency property.

One can ask about degenerate cases of the quad-graph systems that formed the starting-point of the present investigations, in particular this includes all remaining systems from the broader classes in \cite{ABSf,AtkNie}.
The degeneration procedures break the symmetry, essentially reintroducing the distinction between variables and parameters. 
Due to this symmetry breaking, the extension of the limiting procedure from the hypercube, which is now considered as a sub-lattice, to the extended space, is non-unique.
For this reason it is an interesting task to find what is the natural manifestation of the extended consistency property for the degenerate systems.
A similar situation is also encountered if one considers systems related to the original quad-graph equations by non-local B\"acklund or Miura-type transformations.
Several such transformations are known \cite{PTV,PSTV,KaNie,KaNie3}, and they of course have some natural extension to the hypercube.
Again however, it is non-trivial to extend such transformations beyond the hypercube to the new lattice, and in general one can expect that the additional structure encoded in the extended lattice will become non-local or be otherwise obscured.

We have emphasised strongly a natural algebraic setting for the considered systems, namely the discrete part of the theory of elliptic functions.
The connection is through a class of biquadratic polynomials that we call idempotent (Definition \ref{idempotent}), which are connected with period trisection, and from which the considered systems emerge as the three-orbit constraint.
This gives a perspective from which applications and further generalisations of these systems may arise.
For instance, it reveals a further hidden symmetry of the studied systems that puts the fixed parameters (associated with the constant third orbit) on the same footing as the variables.
Combined with the completeness principle, this hidden symmetry suggests there may exist a further extended consistency property and lattice in which there is no distinguished orbit.
Obtaining such extension and asking about its relation to integrability is the most obvious further question raised by this perspective.

\appendix
\section{Idempotent biquadratic polynomials in the Weierstrass theory of elliptic functions}\label{H}
Consider the symmetric biquadratic polynomial,
\begin{equation}
h(x,y) = c_0+c_1(x+y)+c_2xy+c_3(x^2+y^2)+c_4xy(x+y)+c_5x^2y^2, \label{sb2}
\end{equation}
$c_0,\ldots,c_5\in\mathbb{C}$.
As mentioned in the main text, its orbits (\ref{orbit}), (\ref{orbdef}) may be taken as a definition of elliptic functions in the discrete setting.
We develop here an aspect of this discrete theory connected to the Weierstrass addition formula, and relate it to the idempotent class of biquadratics.

A basic role in this aspect of the Weierstrass theory is played by the discriminant polynomial $r$ associated with $h$, which is defined as
\begin{equation}
\dis[h(x,y),y] = r(x), \quad  \dis :={\textrm{discriminant}}.
\end{equation}
This terminology is unambiguous due to the symmetry of $h$.
Henceforth we assume that $r$ has at least one simple root.
The fundamental property on the level of biquadratics is a kind of closure under composition within the family with shared discriminant polynomial.
More precisely, if biquadratics $h_1$ and $h_2$ have, up to a scalar multiple, a common discriminant polynomial $r$, then the composition of $h_1$ and $h_2$ takes the form
\begin{equation}
\res[h_1(x,y),h_2(y,z),y] = h_3(x,z)h_4(x,z), \quad \res:={\textrm{resultant,}}\label{composition}
\end{equation}
where $h_3$ and $h_4$ are again symmetric biquadratic polynomials whose discriminant polynomial is a (now possibly zero) scalar multiple of $r$.
The appearance of two factors on the right-hand-side of (\ref{composition}) means this composition gives a structure which is sometimes called a two-valued group.

A natural parameterisation of this group is obtained through the evaluation map,
\begin{equation}
h\mapsto \alpha: h(e_0,\alpha)=0.\label{evaluation}
\end{equation}
Choosing the point of evaluation, $e_0\in\mathbb{C}\cup\{\infty\}$, to be the assumed simple root of $r$, makes the evaluation map single-valued.
Furthermore, $h$ can be re-constructed (up to constant scalar multiple) from its discriminant polynomial and the associated parameter $\alpha$.
Specifically
\begin{equation}
h(x,y) \sim H(x,y,\alpha) \label{hH}
\end{equation}
where $H$ is a triquadratic polynomial which may be expressed as a discriminant
\begin{equation}
H := \dis\left[ \frac{r_1(t)r_2(e_0)-r_2(t)r_1(e_0)}{t-e_0}, t\right], \label{bqd}
\end{equation}
where $t$ is a dummy variable, and
\begin{equation}
r_1(t)\sim \frac{r(t)}{t-e_0}, \quad r_2(t)=(t-x)(t-y)(t-\alpha).
\end{equation}
Here $r_1$ is also a polynomial of at most degree three because $e_0$ is the assumed simple root of $r$.
Observe that $H$ is invariant under permutations of the three variables $x$, $y$ and $\alpha$.
This permutation symmetry has the remarkable consequence that the parameters $\alpha_1,\ldots,\alpha_4$ obtained by the evaluation map from the biquadratics $h_1,\ldots,h_4$ appearing in the composition formula (\ref{composition}), 
\begin{equation}
h_i\mapsto\alpha_i:h_i(e_0,\alpha_i)=0 \quad \Rightarrow \quad h_i(x,y)\sim H(x,y,\alpha_i),\quad i\in\{1,2,3,4\},
\end{equation}
are themselves related by the same polynomial $H$,
\begin{equation}
H(\alpha_1,\alpha_2,z)=0 \ \Leftrightarrow \ z \in\{\alpha_3,\alpha_4\}.
\end{equation}
In particular, the binary operation $(x,y)\mapsto z$ defined by equation $H(x,y,z)=0$ endows $\mathbb{C}\cup\{\infty\}$ with the structure of an abelian two-valued group.
Most of the important features of $H$ (\ref{bqd}) can be seen from its canonical form, the case $e_0=\infty$, $r(t)=4t^3-g_2t-g_3$:
\begin{equation}
H(x,y,z) = (xy+yz+zx+\tfrac{1}{4}g_2)^2-4(x+y+z)(xyz-\tfrac{1}{4}g_3).\label{wb}
\end{equation}
Those features being that it appears as the discriminant of a polynomial which is degree-one in each of three variables, that it is symmetric in the three variables, and, through the Weierstrass addition formula
\begin{equation}
H(\wp(\alpha),\wp(\beta),\wp(\gamma))=0 \quad \Leftrightarrow \quad \wp(\gamma) = \wp(\alpha\pm\beta),\label{af}
\end{equation}
also that it defines a two-valued abelian group structure.
The more generic form (\ref{bqd}) reveals the connection with the Caley-Bezout-type expression for the idempotent biquadratic correspondences (\ref{redrag}).

\begin{acknowledgements}
This research was funded by Australian Research Council Discovery Grant DP 110104151.
Thanks are due to Pavlos Kassotakis for many discussions and David Howden for advice on Magma.
\end{acknowledgements}

\bibliographystyle{unsrt}
\bibliography{references}
\end{document}